%% file: main.tex
\documentclass[11pt]{article}
\usepackage{graphicx} %
\usepackage[margin=1in]{geometry}
\input{macros}

\title{Mildly-Interacting Fermionic Unitaries are Efficiently Learnable}
\author{Vishnu Iyer\thanks{University of Texas at Austin, \texttt{vishnu.iyer@utexas.edu}}}
\date{\today}

\begin{document}

\maketitle

\begin{abstract}
    Recent work has shown that one can efficiently learn fermionic Gaussian unitaries, also commonly known as nearest-neighbor matchcircuits or non-interacting fermionic unitaries.
    However, one could ask a similar question about unitaries that are near Gaussian: for example, unitaries prepared with a small number of non-Gaussian circuit elements. These operators find significance in quantum chemistry and many-body physics, yet no algorithm exists to learn them.
    
    We give the first such result by devising an algorithm which makes queries to an $n$-mode fermionic unitary $\inputu$ prepared by at most $O(t)$ non-Gaussian gates and returns a circuit approximating $\inputu$ to diamond distance $\eps$ in time $\poly(n,2^t,1/\eps)$. This resolves a central open question of Mele and Herasymenko under the strongest distance metric \cite{mele2024efficient}. In fact, our algorithm is much more general: we define a property of unitary Gaussianity known as unitary Gaussian dimension and show that our algorithm can learn $n$-mode unitaries of Gaussian dimension at least $2n - O(t)$ in time $\poly(n,2^t,1/\eps)$. Indeed, this class subsumes unitaries prepared by at most $O(t)$ non-Gaussian gates but also includes several unitaries that require up to $2^{O(t)}$ non-Gaussian gates to construct.
    
    In addition, we give a $\poly(n,1/\eps)$-time algorithm to distinguish
    whether an $n$-mode unitary is of Gaussian dimension at least $k$ or $\eps$-far from all such unitaries in Frobenius distance, promised that one is the case.
    Along the way, we prove structural results about near-Gaussian fermionic unitaries  that are likely to be of independent interest.
\end{abstract}

\clearpage

\section{Introduction} \label{sec:intro}

The study of quantum learning theory has seen rapid advancements in recent years, driven by the growing intersection of quantum information science, machine learning, and computational complexity. 
One particularly important task in this domain, known as quantum process tomography, concerns learning unitary transformations of quantum states efficiently \cite{mohseni2008quantum}. It's not difficult to imagine why efficient process tomography could be immensely useful. For example, querying and learning an unknown process observed in nature could lead to an efficient implementation and subsequent simulation of this process on a quantum computer.

Unfortunately, learning unitary operators is prohibitively expensive in general. There exist several sample and time complexity lower bounds that render learning arbitrary unitary processes intractable \cite{baldwin2014quantum,haah2023query,zhao2024learning}. For this reason, when computational efficiency is desired, attention is restricted to simpler, more structured classes of operators. 

One such class is fermionic Gaussian unitaries, also referred to collectively as fermionic linear optics (FLO). Gaussian unitaries are generated by quadratic Hamiltonians in fermionic creation and annihilation operators and play a critical role in condensed matter physics, quantum chemistry, many-body physics, and quantum simulation \cite{echenique2007mathematical, boutin2021quantum, carollo2018symmetric}. In particular, Gaussian states and operators can be used to model systems where fermions (most commonly electrons) are non-interacting. Such objects correspond to fermionic Hamiltonians that are quadratic in the Majorana basis.
This approximation simplifies simulation and optimization tasks while in many cases maintaining a reasonable level of accuracy \cite{bravyi2019approximation}. FLO is also equivalent to another well-studied class of quantum circuits known as matchgates, and both FLO and matchgates admit efficient classical simulation algorithms, expanding their utility further \cite{valiant2001quantum,terhal2002classical,knill2001fermionic}. Furthermore, several computationally efficient learning algorithms are known for Gaussian states and unitaries \cite{oszmaniec2020fermion,ogorman2022fermionic,aaronson2023efficient,mele2024efficient}.

Crucially, however, FLO is not universal for quantum computation \cite{divincenzo2005fermionic}. Of course, this should come as no surprise, given the efficiency of classical simulations. The addition of non-Gaussian circuit elements, often termed doping, boosts this class to universality \cite{jozsa2008matchgates,brod2011matchgates}.
At a high level, states and unitaries with roughly $t$ non-Gaussian elements are referred to as \say{$t$-doped.} 
Just like Gaussian objects correspond to non-interacting fermionic systems, $t$-doped Gaussian objects correspond to what we call \emph{mildly-interacting} fermionic systems. They are associated with fermionic Hamiltonians that include a small number of non-quadratic Majorana terms, each of which induces an interaction between multiple fermions.

Many physically and computationally relevant unitary operators can be expressed as $t$-doped Gaussian objects. For instance, approximations of quantum impurity models often rely on states and operations that are non-Gaussian and require additional elements to describe \cite{bravyi2017impurity,boutin2021quantum}.
As such, unitaries induced by mildly-interacting fermionic Hamiltonians are highly relevant in many-body physics.
Furthermore, quantum chemistry computations beyond the Hartree-Fock method that require even minor interactions between electrons include doped Gaussian processes \cite{jones2021hartree}. Such processes could be critical in developing simulations of more general quantum chemical phenomena.
Understanding how to efficiently learn mildly-interacting fermionic operations is crucial in advancing our capabilities in these applications.

The problem of learning $t$-doped Gaussian states has been resolved. Mele and Herasymenko give an algorithm for learning a classical description of a $t$-doped Gaussian state which runs in time $\poly(n, 2^t)$ \cite{mele2024efficient}. Similar problems have also been studied in the somewhat related context of $t$-doped stabilizer states \cite{grewal2024efficient,leone2024learning,hangleiter2024bell} and $t$-doped bosonic Gaussian states \cite{mele2024continuous}. However, we underscore that the problems of learning states and learning unitaries are incomparable. While the access model can be considered stronger in the unitary learning task, so are the output requirements. Indeed, descriptions of unitary operators encode actions on exponentially many states, and if the goal is to learn the input unitary to low diamond distance error, the algorithm must output a description that is close even on a \emph{worst case} state.

Indeed, the problem of learning doped unitary operators is not nearly as well-studied. Leone, Oliviero, Lloyd, and Hamma give an algorithm to learn Clifford circuits doped specifically with $T$ gates \cite{leone2024decoders}, but this algorithm does not generalize to other non-Clifford gates.
Furthermore, there are unique challenges in the doped Gaussian case: most notably the non-discrete nature of Gaussian operations. Indeed, the algorithm to learn Clifford+$T$ circuits makes specific use of both the discrete nature of Clifford circuits and the specific structure of $T$-gates, which are essentially maximally far from Clifford among single-qubit unitaries.
In their work on learning $t$-doped Gaussian states, Mele and Herasymenko thus posed the problem of learning $t$-doped Gaussian unitaries as a major open question \cite{mele2024efficient}. We resolve this question by giving an efficient learning algorithm for a class of unitary operators that subsumes $t$-doped Gaussian unitaries.

\subsection{Main Results}\label{sec:results}
We characterize the learnability of a fermionic unitary by defining the unitary Gaussian dimension, which quantifies the Gaussianity of a unitary operator. Recall that a Gaussian unitary on $n$ modes conjugates the Majorana operators $\majorana_1,...\majorana_{2n}$ each into linear combinations of Majorana operators. The latter are commonly known as \say{rotated} Majorana operators. Generalizing this notion, we say a $n$-mode fermionic unitary $U$ has Gaussian dimension at least $k$ if there exist $k$ mutually orthogonal rotated Majorana operators that $U$ conjugates into rotated Majorana operators. 

The Gaussian dimension is an integer between $0$ and $2n$, inclusive, and, naturally, is maximized for Gaussian unitaries. The class of $n$-mode Gaussian unitaries of dimension at least $2n-2\kappa t$ subsumes the class of unitaries prepared by at most $t$ $\kappa$-local (arbitrary) non-Gaussian gates. By the Jordan-Wigner transform, the latter is equivalent to the class of matchgate circuits with at most $t$ $\kappa$-local non-matchgate elements.

\subsubsection*{Property Testing}
As a first step towards a learning algorithm, we show that the Gaussian dimension can be efficiently property tested.

\begin{theorem}
    There exists a $\poly(n, 1/\eps, \log(1/\delta))$-time algorithm which makes queries to a fermionic unitary $U$ on $n$ modes and, promised that either:
    \begin{enumerate}[label=(\arabic*)]
        \item $U$ has Gaussian dimension at least $k$
        \item $U$ is $\eps$-far in Frobenius distance from all such unitaries
    \end{enumerate}
    determines which is the case with probability at least $1 - \delta$.
\end{theorem}

The Frobenius distance condition above is crucial for computational efficiency. Indeed, distinguishing unitaries of Gaussian dimension $k$ from unitaries that are $\eps$-far in diamond distance from all such unitaries would imply an algorithm for diamond-distance identity testing. This is a notoriously difficult property testing problem, requiring $2^{\Omega(n)}$ sample complexity and rendering it intractable in most cases \cite{rosgen2004hardness}.

In the case of $k=2n$, our algorithm gives a computationally efficient property tester for Gaussian unitaries, resolving an open question of Bittel, Mele, Eisert, and Leone \cite{bittel2025fermion}. We clarify, however, that the sample complexity of this algorithm is also $\poly(n)$ (as opposed to constant). It is unclear whether a sample complexity growing with $n$ is necessary to test for unitary Gaussianity, though we conjecture that it is not.

\subsubsection*{Efficient Learning}
Our central result is a learning algorithm for $n$-mode unitaries of Gaussian dimension at least $2n-t$ that scales polynomially in $n$ and exponentially in $t$.

\begin{theorem} \label{thm:main-learning-intro}
    There exists a $\poly(n, 2^t, 1/\eps, \log(1/\delta))$-time algorithm that makes queries to an $n$-mode fermionic unitary $U$ with Gaussian dimension at least $2n-t$ and produces a circuit that implements a unitary $\wt{U}$ such that $\diamonddistance{U}{\wt{U}} \leq \eps$, succeeding with probability at least $1-\delta$.
\end{theorem}

By the Jordan-Wigner transformation, our algorithm is also capable of learning matchgate circuits augmented with a few non-matchgate circuit elements.

A priori, we believe this result is rather surprising. Gaussian unitaries themselves form a continuous class of objects, and the result of \cite{oszmaniec2020fermion} that learns them to small diamond distance relies on a concise projective representation in the form of group $\orthgrp(2n)$. Such a representation is not known to exist for unitaries of high Gaussian dimension. It's not clear at all that the continuous variables of their concise representation could be efficiently approximated to a level that allows for learning to low error in diamond distance, which quantifies the \emph{worst case} distance between two unitaries. Indeed, no such algorithm has been developed for, say, doped bosonic Gaussian or Clifford unitaries with arbitrary non-Clifford gates, and we believe our result comprises new techniques that can be applied to give novel learning algorithms for these classes.

This algorithm is effectively a \emph{proper} learning algorithm: the unitary $\wt{U}$ implemented is itself of Gaussian dimension at least $2 \floor{k/2}$.
Thus the Gaussian dimension of the input unitary is preserved if $k$ is even but if $k$ is odd then the output unitary could have Gaussian dimension $k-1$. 
It is worth clarifying, however, that the algorithm is not proper for the class of unitaries with at most $t$ $\kappa$-local non-Gaussian gates. While the algorithm remains efficient when $\kappa t$ is small, the circuit produced could contain a number of non-Gaussian gates that is exponential in $\kappa t$.

While the fine-grained sample and time complexities of the algorithm could certainly be improved, a qualitative $\poly(n, 2^t, 1/\eps)$ runtime is the best one can do. To see this, consider the class of $n$-mode unitaries that consist of an arbitrary Gaussian unitary on the first $2n-t$ modes and an arbitrary unitary on the last $t$ modes.
Clearly, every unitary in this class has Gaussian dimension at least $2n-t$.
Learning the Gaussian unitary over the first register requires $\poly(2n-t)$ sample complexity by the Holevo bound and learning the arbitrary unitary over the second register requires $\exp(t)$ sample complexity \cite{baldwin2014quantum,haah2023query}. In fact, this scaling is required even in the average case:
a standard counting argument shows that an overwhelming fraction of unitaries in this class have gate complexity at least $\poly(n, 2^t/t)$. Applying results such as the lower bounds in \cite{zhao2024learning}, one can argue that a $\poly(n, 2^t, 1/\eps)$ scaling is essentially necessary to learn this class of operators to Frobenius distance within $\eps$.

\subsection{Technical Overview} \label{sec:technical-overview}
The technical machinery behind our results 
involves the representation of fermionic systems via Majorana operators. 
A good summary can be found in \cite{bravyi2004lagrangian}, though we review the key concepts.

For any fermionic system on $n$ modes, we can define the \emph{Majorana operators} $\majorana_1,\majorana_2,...,\majorana_{2n}$. These operators can be thought of as $2^n \times 2^n$ matrices analogous to Pauli operators. In fact, Majorana operators and their products are in correspondence with Pauli operators via the Jordan-Wigner transformation.
The Majorana operators satisfy the so-called Fermi-Dirac commutation relations: 
\[
\{\majorana_a, \majorana_b\} = 2\delta_{ab}.
\] 
This condition implies an orthogonality relation under the Hilbert-Schmidt inner product: $\trace{\majorana_a \majorana_b} = 2^n \delta_{ab}$. Just like any set of generators for the Pauli group, the Majorana operators generate an algebra which spans all $2^n \times 2^n$ matrices. That is, any matrix in $\C^{2^n \times 2^n}$ can be written as a linear combination of products of the $\majorana_a$.

For any fermionic unitary $U$ on $n$ modes, we can define a
\emph{correlation matrix} $M$, where $M_{ab} = 2^{-n}\trace{\majorana_a\cdot U \majorana_b U^\dagger}$.
Observe that any two distinct rows of the correlation matrix are orthogonal, since 
\[
2^{-n}\trace{U \majorana_a U^\dagger \cdot U \majorana_b U^\dagger} = 2^{-n}\trace{\majorana_a \majorana_b} = \delta_{ab}.
\]
If $U$ is Gaussian, the correlation matrix $M$ lies in $\orthgrp(2n)$ and forms a complete description of the action of $U$. That is, we can write
\[
U\majorana_a U^\dagger = \sum_{b=1}^{2n} M_{ba} \majorana_b = \newmajorana_a.
\]
Thus, the $\eta_a$ formed by conjugation of the Majorana operators by Gaussian unitaries are linear combinations of Majorana operators, and we will refer to these as \emph{rotated} Majorana operators. Oszmaniec, Dangniam, Morales, and Zimbor\'as showed that two Gaussian operators whose correlation matrices are close also have small diamond distance \cite{oszmaniec2020fermion}. Thus, approximating fermionic Gaussian unitaries reduces to approximating their correlation matrices sufficiently well.
The powerful structure of the correlation matrix for Gaussian operators begs the question: what happens for arbitrary unitaries?

A central observation in this work is that if the unitary is \say{somewhat Gaussian,} then the correlation matrix (and particularly its singular value decomposition) should still be amply informative. To formalize this notion of near-Gaussianity, we introduce the unitary \emph{Gaussian dimension}, which quantifies, as an integer between $0$ and $2n$, the number of rotated Majorana operators that the unitary conjugates into other rotated Majorana operators. The Gaussian dimension was defined for states in \cite{mele2024efficient}, who showed how to learn fermionic states of small Gaussian dimension efficiently. In this work, the term Gaussian dimension will specifically refer to unitary Gaussian dimension. A formal definition of unitary Gaussian dimension is given in \cref{def:gaussian-dimension}. 

Leveraging this definition, we show that any unitary on $n$ modes with Gaussian dimension at least $2n-t$ can be expressed in the following form:
\[
U = \leftgaussian (\identity \otimes u) \rightgaussian^\dagger,
\]
where $\leftgaussian$ and $\rightgaussian$ are Gaussian unitaries and $u$ acts only on the last $\ceil{t/2}$ modes (see \cref{lem:compression}).
This generalizes the compression lemma showed by Mele and Herasymenko in \cite{mele2024efficient}, since the class of high Gaussian dimension unitaries subsumes the class of unitaries prepared with a small number of non-Gaussian gates.  When referring to $\leftgaussian$ and $\rightgaussian$, we will often say they \emph{diagonalize} $U$.

The astute reader may notice that the form of the so-called \say{compression lemma} above resembles that of the singular value decomposition. This intuition turns out to be critical in exploiting the structure of the correlation matrix.
Indeed, unitaries whose correlation matrices have exactly $k$ singular values equal to $1$ are exactly those unitaries of Gaussian dimension $k$. Essentially, the unitary conjugates the rotated Majorana operators represented by the first $k$ right-singular vectors of $M$ into those represented by the first $k$ left-singular vectors. We illustrate this in \Cref{fig:circuit-vs-correlation}.

\begin{figure}[!ht]
    \centering
    \includegraphics[width=3in]{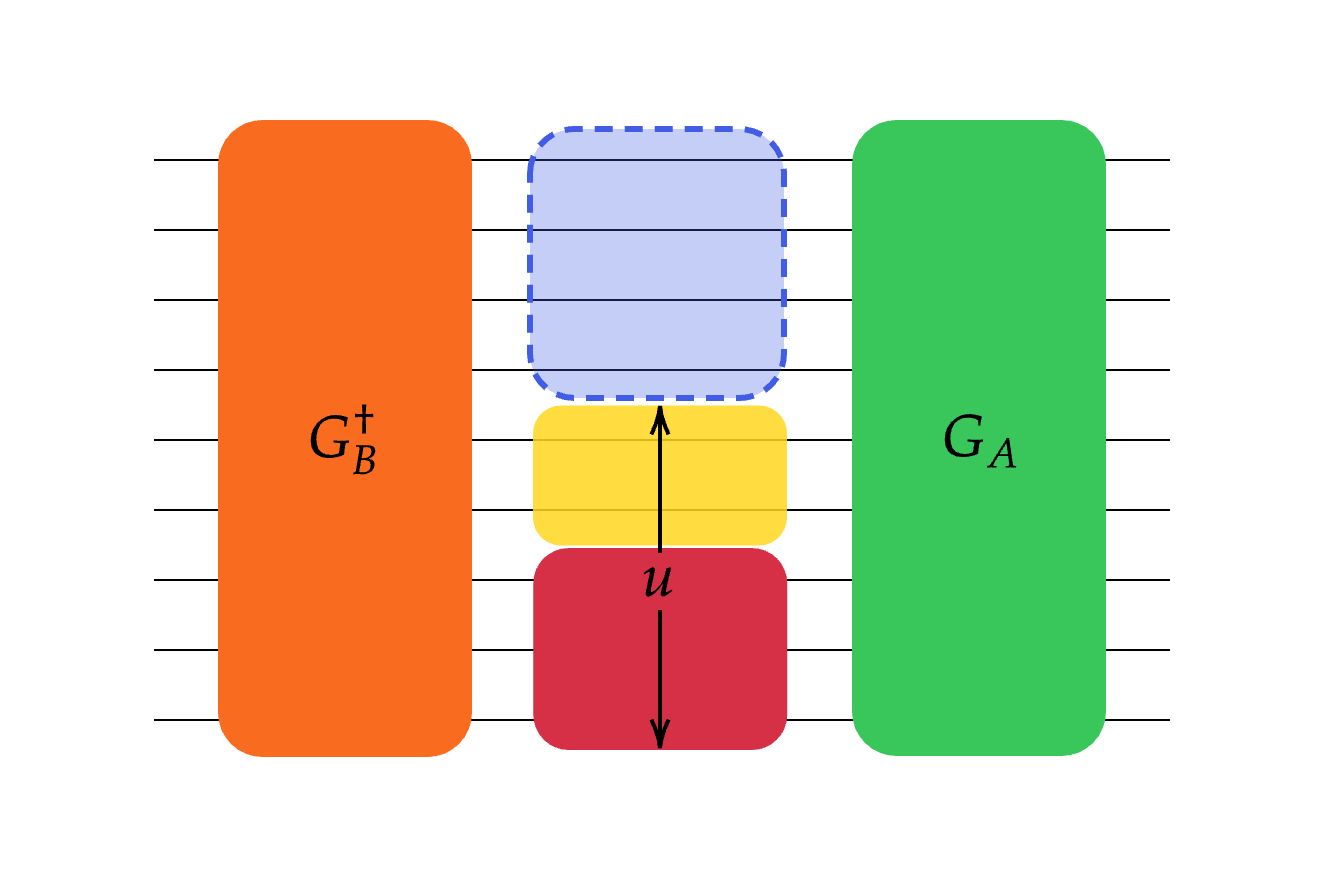} \\
    \includegraphics[width=3in]{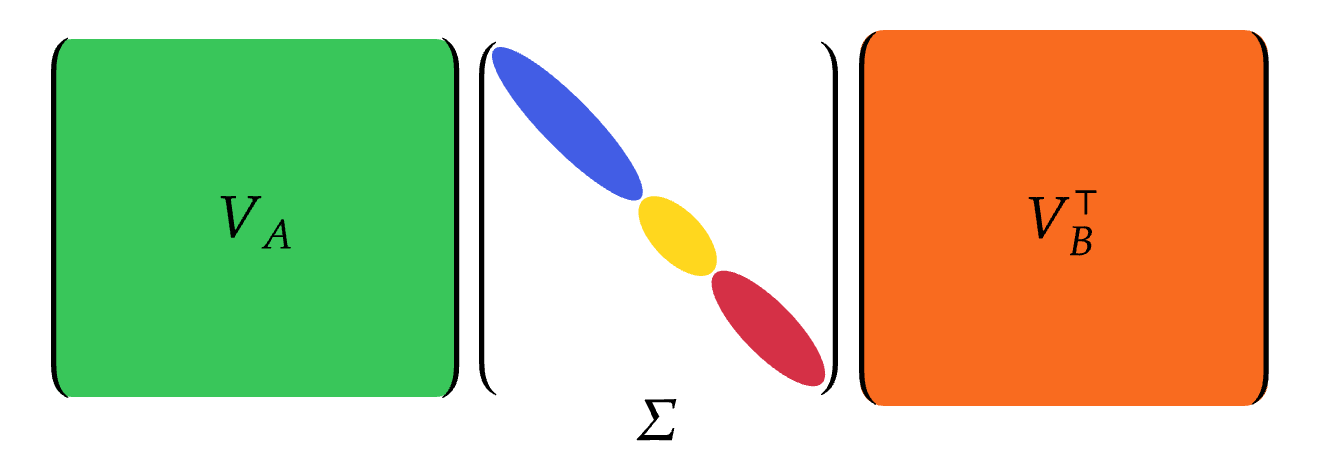}
    \caption{The correspondence between the circuit of a $9$-mode unitary with Gaussian dimension $8$ and the singular value decomposition of its correlation matrix (which is $18 \times 18$). Related components (such as $V_A$ and $\leftgaussian$) have the same color. The first eight singular values (marked \textbf{\color{figBlue}blue}) are equal to $1$, so the middle layer of the circuit acts as the identity on the first four modes. On the other hand, the last six singular values (marked \textbf{\color{figRed}red}) are small, so they correspond to a non-Gaussian component of the middle layer. The middle four eigenvalues (marked \textbf{\color{figYellowText}yellow}) are close to $1$, and care must be taken to determine whether the corresponding modes in the middle layer are treated as Gaussian or non-Gaussian.} 
    \label{fig:circuit-vs-correlation}
\end{figure}

With this structure in mind, the high level approach to our learning algorithm is:
\begin{enumerate}
    \item Construct a sufficiently accurate estimate of the correlation matrix $M$ and compute its singular value decomposition $V_A \Sigma V_B^\transpose$.
    \item We will determine $\pk$, the index of the lowest singular value that is sufficiently close to $1$. This alone will suffice to solve solve the problem of testing Gaussian dimension at least $k$, as we can simply check that $\pk \geq k$.
    
    \item Using the $V_A$ and $V_B$ (which are in $\orthgrp(2n)$) we define the Gaussian operators $\estleftgaussian$ and $\estrightgaussian$, respectively.
    
    \item To learn the non-Gaussian component $\inneru$, we will perform brute-force tomography on the final $n - \floor{\pk/2}$ modes of $\estleftgaussian U \estrightgaussian^\dagger $, approximating this channel in diamond distance. 
    
    \item We can then round this estimate to a unitary channel $\innerpolaru$ by taking, for example, its polar decomposition. A circuit for a unitary $\innercircuitu$ approximating $\innerpolaru$ can be constructed by standard techniques, such as \cite{dawson2005solovay}.
    \item The final output of the algorithm will be a circuit for $\estrightgaussian^\dagger$, followed by the circuit for $\innercircuitu$ acting on the final $n - \floor{\pk/2}$ modes, then finally by a circuit for $\estleftgaussian$. This implements the unitary $\estleftgaussian (\identity \otimes \innercircuitu) \estrightgaussian^\dagger$.
\end{enumerate}

For the property testing algorithm, it will suffice to analyze singular values of the correlation matrix.
We show that unitaries whose correlation matrices have many large singular values can be rounded to unitaries of high Gaussian dimension.
On the contrapositive, this tells us that unitaries that are far from having Gaussian dimension $k$ have correlation matrices whose $k$-th singular value is bounded away from $1$. On the other hand, unitaries of high Gaussian dimension have correlation matrices with several singular values equal to $1$. As such, the proof of correctness hinges on the accuracy of our estimate of $M$, which can be argued using perturbation bounds such as Weyl's theorem (\cref{fact:weyl}).

The learning algorithm requires a bit more work. One key challenge is that, just like the singular value decomposition isn't unique, neither is the form $\leftgaussian (\identity \otimes u) \rightgaussian^\dagger$. Does it matter which decomposition we find?
We will argue that the decomposition will not actually matter, so long as the \emph{subspaces} spanned by the top $k$ singular vectors of $M$ and $\wh{M}$ are close. This also implies that the subspaces spanned by the bottom $2n-k$ singular vectors are close. Under this condition, for any pair $\estleftgaussian, \estrightgaussian$ produced from $V_A,V_B$, there exists a pair $\leftgaussian,\rightgaussian$ such that $\diamonddistance{\leftgaussian}{\estleftgaussian}$ and $\diamonddistance{\rightgaussian}{\estrightgaussian}$ are small and $\leftgaussian^\dagger U \rightgaussian = \identity \otimes u$ for some $u$ on $n - \floor{k/2}$ modes.
As such, $\estleftgaussian^\dagger U \estrightgaussian$ approximately acts as the identity on the first $\floor{k/2}$ modes.

However, just because two matrices are close pointwise or in operator norm does not imply that the subspaces spanned by the singular vectors are close.
To this end, we will need to use the powerful Davis-Kahan theorem (\cref{fact:davis-kahan}), which concerns the stability of the singular value decomposition. This theorem states that the subspace spanned by the top $k$ singular vectors of two matrices $M$ and $\wh{M}$ are close if the matrices are close in operator distance and there is an \say{eigengap} between the $k$-th singular value of $M$ and the $k+1$-th singular value of $\wh{M}$. 

 Even if $M$ has Gaussian dimension exactly $k$ (and not $k+1$) and $\wh{M}$ is very close to $M$ in operator distance, this eigengap condition might not be satisfied. For example, the $k+1$-th singular value of $M$ could be very close to $1$, causing the Davis-Kahan theorem to fail. 
 To avoid this, the algorithm employs a singular value \say{partitioning} step. For every singular value after the $k$-th one, the algorithm applies a series of cutoffs, each lower than the last. Singular values that pass the check will be considered \say{good} and the first singular value that fails, along with all that succeed it, are considered \say{bad.} The final threshold is chosen such that, if all singular values are marked as good, $U$ is close to a fully Gaussian unitary in diamond distance. Decreasing the thresholds at each step ensures that the last good singular value and the first bad singular value will differ by at least the gap between the thresholds, inducing an eigengap. 
 
 We can then argue that, rounding all good eigenvalues to $1$, $U$ is close in diamond distance to a unitary $U^\prime$ of higher Gaussian dimension $\pk \geq k$ that is diagonalized by the same $\leftgaussian, \rightgaussian$. That is, there exists some $u^\prime$ acting on $n - \floor{\pk/2}$ modes such that $U^\prime = G_A (\identity \otimes u) G_B^\dagger$. This \say{rounding} is depicted in \Cref{fig:circuit-rounding}. The eigengap condition allows us to effectively approximate $\leftgaussian$ and $\rightgaussian$ with $\estleftgaussian$ and $\estrightgaussian$. By considering $\estleftgaussian^\dagger U \estrightgaussian$ and tracing out the first $\pk$ modes, we construct an approximation $\innercircuitu$ of $\innerroundedu$. 
The correctness of the algorithm ultimately follows from showing that $\outputu = \estleftgaussian (\identity \otimes \innercircuitu) \estrightgaussian^\dagger$ is close to $\inputu = \leftgaussian(\identity \otimes u) \rightgaussian^\dagger$ via an intuitive, albeit involved, hybrid argument.

\begin{figure}
    \centering
    \includegraphics[width=3in]{diagrams/u-circuit.png}
    \includegraphics[width=3in]{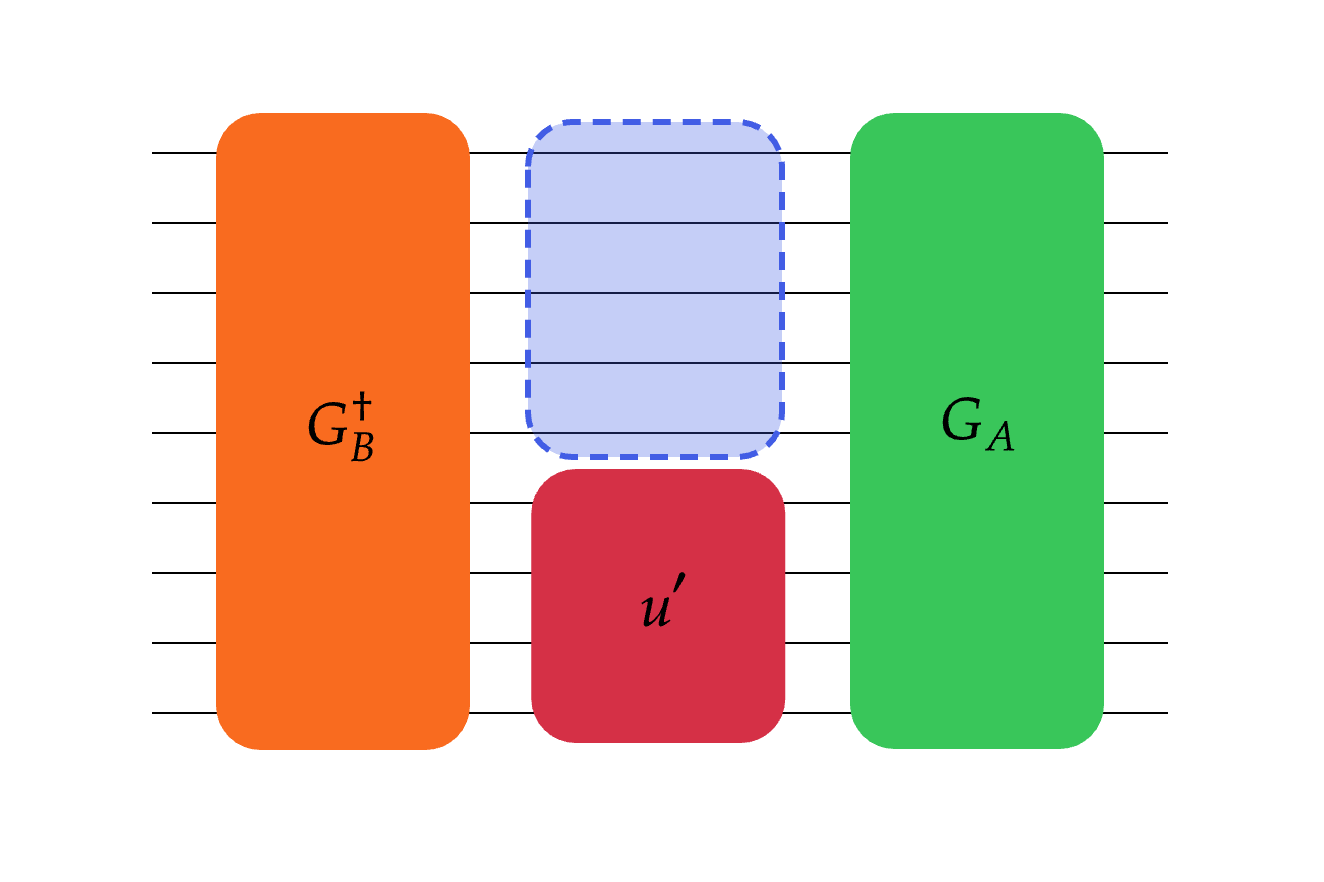}
    \caption{A depiction of the ``rounding'' that implicitly takes place in the algorithm. On the left, we have the same circuit as in \Cref{fig:circuit-vs-correlation}. This circuit could be close to having higher Gaussian dimension, since the $9$th to $12$th singular values (corresponding to the \textbf{\color{figYellowText}yellow} region) are close to $1$. Based on a sequence of thresholds, we determine that the $9$th and $10$th singular values pass the checks whereas the singular values from $11$ onward fail. This will allow us to argue that the circuit is close in diamond distance to the one on the right, which has Gaussian dimension $10$ and is diagonalized by the same Gaussian operators. The learning algorithm then proceeds by approximating $\leftgaussian,\rightgaussian,$ and $\innerroundedu$.} 
    \label{fig:circuit-rounding}
\end{figure}

Overall, we remark that the techniques employed in this result bear some resemblance to those in algorithms for learning objects such as $t$-doped stabilizer states and unitaries \cite{leone2024decoders,grewal2024improved,grewal2024efficient,leone2024learning,hangleiter2024bell}, as well as those for learning $t$-doped Gaussian states \cite{mele2024efficient}. At a high level, we prove a compression lemma, \say{unravel} the structure of the compression, and perform brute-force tomography on the compressed unitary. However, the setting of fermionic unitaries and the goal of small diamond error approximation induced major hurdles and thus required new techniques.
We assert that many of these techniques should be capable of strengthening the learning results of other \say{compressible} classes of unitary operators.

\subsection{Open Questions}

As mentioned previously, our learning algorithm is, in some sense, the first of its kind to learn doped unitaries to small diamond distance. We see no barrier to applying these techniques to the settings of doped Clifford and bosonic Gaussian circuits, even if each case presents unique challenges.
\begin{question*}
    Do there exist algorithms to learn $t$-doped $n$-qubit Clifford circuits and $t$-doped $n$-mode bosonic Gaussian circuits\footnote{Just as in \cite{mele2024continuous}, there is likely to be an energy constraint on how much the unitary can amplify the energy of an input state.}
    in $\poly(n, 2^t)$ time?
\end{question*}
Both of these classes exhibit a compressive structure like the one we've exploited to learn doped fermionic Gaussian unitaries \cite{leone2024decoders,bittel2025boson}, so it's our strong belief that similar techniques will yield promising results.

At a high level, our work explores the complexity of learning fermionic operators as a function of their \emph{non-Gaussianity}. In a similar vein, one can pose the problem of \emph{agnostic} learning of Gaussian unitaries. Specifically, given a unitary $U$ that is $\eps$-close to some Gaussian unitary in Frobenius distance, can one find a Gaussian unitary that is $(\eps+\eps^\prime)$-close to $U$? Similar notions have been explored in the learning of states \cite{grewal2024agnostic,chen2024stabilizer,bakshi2024learning}, but not much is known about agnostic tomography of unitaries. Some progress was made in \cite{wadhwa2024agnostic}, but not in the proper case, where a circuit in the class needs to be returned by the algorithm. 

\begin{question*}
    Does there exist a proper agnostic tomography algorithm for Gaussian unitaries?
\end{question*}
Agnostic process tomography is very well motivated by considerations of gate and measurement errors, and such an algorithm for Gaussian unitaries would certainly be valuable. 

Similarly, one could ask about the analogous question of tolerant testing: given oracle access to a Gaussian unitary or state, is there an efficient procedure to test whether it is $\eps_1$-close to the class of Gaussian objects or $\eps_2$-far from all such objects? Similar questions have been studied in the world of stabilizer states, so perhaps those techniques could be fruitful in the Gaussian setting \cite{grewal2024improved,arunachalam2024tolerant,bao2024tolerant,iyer2024tolerant}.
\begin{question*}
    Do there exist computationally efficient tolerant property testers for Gaussian states and unitaries?
\end{question*}
The standard (non-tolerant) property testing question has been studied in \cite{mele2024efficient,lyu2024fermionic,bittel2025fermion}, but it's not clear whether the techniques in these works can be used to produce tolerant testers. Furthermore, all of these tests have sample complexities that depend on $n$. We conjecture that there exists a property testing algorithm that can property test Gaussian unitaries with only constantly many queries.

\begin{question*}
    Does there exist an algorithm to test whether a $n$-mode unitary $U$ is either a Gaussian unitary or $\eps$-far in Frobenius distance from all such unitaries
    given $\poly(1/\eps)$ queries to $U$ (rather than $\poly(n, 1/\eps)$) and $\poly(n, 1/\eps)$ runtime?
\end{question*}

Of course, we can also ask if the polynomial factors in $n, \eps,$ and $2^t$ in the algorithm of \cref{thm:main-learning-intro} can be improved. For example, the runtime scales as $\poly(n) \cdot \poly(2^t)$, and one could ask if the dependence on $n$ and $t$ could be decoupled into something like $\poly(n) + \poly(2^t)$.

\section{Preliminaries} \label{sec:prelims}

$\norm{M}$ denotes the operator norm of matrix $M$ and $\frobeniusnorm{M}$ denotes the Frobenius norm. The set $\{1,...,n\}$ is denoted by $[n]$. $\identity_{k}$ represents the $2^k \times 2^k$ identity operator. The group of orthogonal $n \times n$ matrices is denoted $\orthgrp(n)$.
We denote the $k$-th singular value of an $n\times n$ matrix $M$ by $\sigma_k(M)$. Unless stated otherwise, we assume that the singular values are in sorted order. That is, $\sigma_1(M) \geq ... \geq \sigma_n(M)$. Similarly, we will denote the eigenvalues of $M$ as $\lambda_1(M) \geq ... \geq \lambda_{n}(M)$.
We will often take $\sigma_{a}(M) = \lambda_{a}(M) = -\infty$ for $a > n$.

\subsection{Quantum Information}

We review a few basic facts about quantum information.
The (phase-invariant) Frobenius distance between two $d \times d$ operators is given by
\[
\frobeniusdistance{U,V} = \frac{1}{\sqrt{2d}}\min_{\theta \in [0,2\pi)}\frobeniusnorm{U - e^{i\theta}V} =  \sqrt{1 - \frac{1}{d} \abs{\trace{U^\dagger V}}}.
\]
The diamond distance between two $d \times d$ operators is given by
\[
\diamonddistance{U}{V} = \max_{\rho}\norm{(U \otimes \identity_d) \rho - (V \otimes \identity_d) \rho}_1,
\]
where $\rho$ is any density matrix of dimension $2d$. 
The diamond distance can be thought of as a worst-case distance, quantifying the difference in action by two quantum channels on a worst-case input state. On the other hand, the Frobenius distance can be thought of as an average-case distance metric.
As such, the Frobenius distance is always upper bounded by the diamond distance. Conversely, if $U$ and $V$ are again of dimension $d$, we have:
\[
\diamonddistance{U}{V} \leq \sqrt{2d} \cdot \frobeniusdistance{U,V}.
\]

We recall some basic facts about how these distance measures behave upon applying arbitrary quantum channels:

\begin{fact} \label{fact:diamond-distance-contraction}
    Let $\calV, \calW,$ and $\calX$ be quantum channels and let $U$ be a unitary channel. Then, whenever the relevant quantities are well-defined,
    \begin{enumerate}[label=(\alph*)]
        \item $\diamonddistance{U\calV}{U\calW} = \diamonddistance{\calV U}{\calW W} = \diamonddistance{\calV}{\calW}$.
        \item $\frobeniusdistance{U\calV,U\calW} = \frobeniusdistance{\calV U,\calW U} = \frobeniusdistance{\calV,\calW}$.
        \item \(
        \diamonddistance{\calX \calV}{\calX \calW} \leq \diamonddistance{\calV}{\calW}.\)
    \end{enumerate}
    
\end{fact}

The quantum union bound provides a lower bound on the success probability of binary projective measurements. We recall it below.

\begin{fact}[Quantum Union Bound] \label{fact:quantum-union-bound}
    Let $\rho$ be a quantum state and let $P_1,...,P_k$ be $0-1$ projective measurements such that $\trace{P_t \rho} \geq 1 - \eps_t$ for all $1 \leq t \leq k$.
    Suppose we perform all $k$ measurements sequentially.
    Then the probability that all measurements return $1$ simultaneously is at least
    \[
    1 - 4 \sum_{t=1}^k \eps_t
    \]
\end{fact}

A useful primitive in quantum learning theory is full process tomography \cite{baldwin2014quantum,haah2023query}. 
\begin{fact}
    Let $\calU$ be a $d$-dimensional quantum channel. There exists an algorithm which has sample and time complexity $\poly(d,1/\eps,\log(1/\delta))$ and returns a classical description of a channel $\wh{\calU}$ that is $\eps$-close to $\calU$ in diamond norm, succeeding with probability at least $1-\delta$.
\end{fact}
We will refer to this procedure as $(\eps,\delta)$-tomography and apply it as a black-box routine in our learning algorithm.

\subsection{Linear Algebra}
Many of our proofs involve robustness analysis of linear algebraic procedures such as the singular value decomposition. In this section, we review several fundamental facts that will be useful in those computations.

First, we recall Weyl's theorem, which relates the singular values of a matrix to those of a matrix that is perturbed.

\begin{fact}[Weyl's Theorem]\label{fact:weyl}
Let $A, B$ be $m \times m$ matrices. Then for all $k, \ell \geq 1$,
\[
\sigma_{k + \ell - 1}{(A + B)} \leq \sigma_{k}{(A)} + \sigma_{\ell}{(B)}.
\]
\end{fact}

Keeping the same notation as above, we show that matrices that are close in operator norm also have close singular values.

\begin{proposition}\label{prop:singular-value-closeness}
    Let $M, \wh{M} \in \R^{m \times m}$ Then for all $k \in [m]$,
    \[
    \abs{\sigma_k{(M)} - \sigma_k{(\wh{M})}} \leq \norm{M - \wh{M}}.
    \]
\end{proposition}
\begin{proof}
    Let $\Delta = \wh{M} - M$. Invoking \cref{fact:weyl} with $A = M$, $B = \Delta$, and $\ell = 1$, we can see that
    \[
    \sigma_k{(\wh{M})} - \sigma_k{(M)} \leq \sigma_1{(\Delta)}.
    \]
    On the other hand, invoking \cref{fact:weyl} with $A = \wh{M}$, $B = -\Delta$, and $\ell = 1$, we have that
    \[
     \sigma_k{(M)} - \sigma_k{(\wh{M})} \leq \sigma_1{(-\Delta)} = \sigma_1{(\Delta)}.
    \]
    The result is immediate from the fact that $\norm{\Delta} = \sigma_1(\Delta)$.
\end{proof}

A common notion of distance between inner product subspaces is the set of \emph{canonical angles}.
\begin{definition}
    Let $\calE$ and $\calF$ be two dimension $k$ subspaces of an inner product space. The first canonical angle between $\calE$ and $\calF$ is
    \[
    \theta_1 = \cos^{-1}\left(\max_{x \in \calE} \max_{y \in \calF} \abs{\langle x, y \rangle}\right)
    \]
    Take $x_1$ and $y_1$ to be the maximizers above.
    For any $1 \leq \ell \leq k$, the $\ell$-th canonical angle is
    \[
    \theta_{\ell} = \cos^{-1}\left(\max_{\substack{x \in \calE \\  \forall r \in [\ell - 1], \langle x, x_r \rangle = 0}}\ \  \max_{\substack{x \in \calF \\ \forall r \in [\ell - 1], \langle y, y_r \rangle = 0 }} \abs{\langle x, y \rangle}\right)
    \]
\end{definition}
A consequence of the above definition is that, taking $\{v_1,...,v_k\}$ to be an orthonormal basis for $\calE$, there exists an orthonormal basis $\{w_1,...,w_k\}$ for $\calF$ such that
\[
\min_{a \in [k]} \abs{\langle v_a, w_a\rangle} \geq \min_{a \in [k]} \cos \theta_a = \cos \theta_k.
\]

The Davis-Kahan Theorem quantifies the angle between subspaces spanned by the first $k$ eigenvectors of two matrices, one which is slightly perturbed from the other:

\begin{fact}[Davis-Kahan Theorem \cite{davis1970rotation}] \label{fact:davis-kahan}
    Let $A$ and $B$ be two symmetric matrices. Define the subspaces $\calE$ and $\calF$ to be the span of the first $k$ eigenvectors (corresponding to the largest $k$ eigenvalues) of $A$ and $B$, respectively.
    Furthermore, let $\eigengap = \lambda_k(A) - \lambda_{k+1}(B)$ and let $\Theta$ be the diagonal matrix of the $k$ canonical angles between $\calE$ and $\calF$. Then
    \[
    \norm{\sin \Theta} \leq \frac{\norm{A- B}}{\eigengap}
    \]
    In other words, if $\theta$ is the largest of the $k$ canonical angles between $\calE$ and $\calF$,
    \[
    \sin \theta \leq \frac{\norm{A - B}}{\eigengap}.
    \]
\end{fact}
Going forward, we will denote $\theta(\calE, \calF)$ to be the largest canonical angle between subspaces $\calE$ and $\calF$ of equal dimension. Similarly, for two orthogonal matrices $V,W$, we denote by $\theta(V,W)$ the largest angle between two corresponding columns of $V$ and $W$.

Finally, we will need the following concentration inequality for approximating matrices using unbiased estimators.
\begin{fact}[Matrix Bernstein Inequality] \label{fact:matrix-bernstein}
Let $X^{(1)},...,X^{(m)} \in \R^{n \times n}$ be i.i.d. zero-mean matrix variables such that $\norm{X^{(k)}} \leq M$ for all $1 \leq k \leq m$ and for all $a,b \in [n]$, $\Var[X_{ab}^{(1)}] \leq \sigma^2$. Then
\[
\Pr\left[\frac{1}{m}\norm{X^{(1)} + ... + X^{(m)}} \geq \eps\right] \leq 2n \exp\left(-\frac{m\eps^2}{2\left(\sigma^2 n + M \eps\right)}\right).
\]
    
\end{fact}

For our purposes, $M$ and $\sigma$ can be upper bounded by $1$. As such, we will replace the bound above with the term
\[
2n \exp\left(-\frac{m\eps^2}{3n}\right).
\]

Proofs of many of these theorems can be found in matrix analysis texts such as \cite{bhatia1997matrix} or \cite{stewart1990matrix}.

\subsection{Fermions}
Fermions are particles that obey the Pauli exclusion principle: no two particles can simultaneously occupy the same \emph{mode}, or state. Fermionic systems are most commonly represented in terms of \emph{Majorana operators}. We will consider systems with $n$ modes: in this setting, the Majorana operators are $2^n \times 2^n$ Hermitian matrices $\majorana_1,...,\majorana_{2n}$ that obey the so-called Fermi-Dirac commutation relations:
\[
\{\majorana_a, \majorana_b\} = 2\delta_{ab}
\]
Simply put, $\majorana_a^2 = \identity$ and $\majorana_a \majorana_b = -\majorana_b \majorana_a$ if $a \neq b$. This commutation relation also implies that $\trace{\majorana_a \majorana_b} = 2^n \delta_{ab}$.

The Jordan-Wigner transformation relates $n$-qubit systems to fermionic systems on $n$ modes. Under this transformation, the majorana operators can be thought of as analogues of Pauli operators:
\[
\majorana_{2a-1} = Z^{\otimes a-1}\otimes X \otimes I^{\otimes n-a},\quad \majorana_{2a} = Z^{\otimes a-1}\otimes Y \otimes I^{\otimes n-a}.
\]
This transformation can also be performed in the other direction, converting arbitrary Pauli operators to products of Majorana operators.

Recall that from a \say{basis} of $2n$ independent Pauli operators, one can generate the entire Pauli group, which themselves form a basis for all $2^n \times 2^n$ matrices. This idea extends to the Majorana operators and their products. For this reason, we will refer to the set $\{\majorana_1,...,\majorana_{2n}\}$ as a \say{Majorana basis}. Given a matrix $M \in \orthgrp(2n)$, we can define a new set of operators
\[
\newmajorana_a = \sum_{b=1}^{2n}M_{ba} \majorana_b
\]
that can easily be seen to obey the Fermi-Dirac commutation relations $\{\newmajorana_a,\newmajorana_b\} = 2\delta_{ab}$. The $\newmajorana_a$ are often referred to as \emph{rotated} Majorana operators. We will say a set of Majorana operators is \emph{independent} if no operator in the set can be formed from a linear combination of the other operators.
Any size $2n$ set of independent rotated Majorana operators forms a Majorana basis. 

A \emph{fermionic Gaussian unitary} (hereafter simply \say{Gaussian unitary}) on $n$ modes is a unitary $U$ such that for every $a \in [2n]$, we have
\[
U \majorana_a U^\dagger = \sum_{b=1}^{2n} M_{ba} \majorana_b
\]
where $M \in \orthgrp(2n)$. In other words, a Gaussian unitary acts by conjugation on a basis of Majorana operators by rotating them into a different basis of Majorana operators.

Critically, the $2^n \times 2^n$ matrix $U$ can be uniquely associated with a $2n \times 2n$ matrix representation $M \in \orthgrp(2n)$, known as its \emph{correlation matrix}. The conciseness of this representation means the action of Gaussian unitaries on the standard vacuum state are efficient to simulate classically \cite{valiant2001quantum,knill2001fermionic}. This representation is also central in developing efficient learning algorithms for Gaussian unitaries \cite{oszmaniec2020fermion}. In fact, the diamond distance between two Gaussian operators can be related to the operator distance betweeen their corresponding correlation matrices.
\begin{fact}[\cite{oszmaniec2020fermion}, Lemma 6]\label{fact:gaussian-diamond-distance}
    Let $G$ and $\wh{G}$ be two $n$-mode Gaussian operators corresponding to correlation matrices $M, \wh{M} \in \orthgrp(2n)$, respectively. Then
    \[
    \diamonddistance{G}{\wh{G}} \leq 2n \norm{M - \wh{M}}.
    \]
\end{fact}

Of course, we can define a correlation matrix for an arbitrary unitary $U$ on $n$ modes. However, if $U$ is not Gaussian, the correlation matrix will not be in $\orthgrp(2n)$. In particular, at least one of the rows will have norm less than one. This is because any unitary can be written as a linear combination of \say{Majorana monomials} $\prod_{a \in S} \majorana_a$ for $S \subseteq [2n]$, and by the orthonormality of the Majorana operators under the Hilbert-Schmidt inner product, the sum of the squares of these coefficients is equal to $1$. The correlation matrix only includes the degree-$1$ coefficients, so it's clear that all of the rows of the correlation matrix will have norm at most $1$. A consequence of this fact is that the singular values of the correlation of any $n$-mode fermionic unitary are between $0$ and $1$, with only Gaussian unitaries having $n$ singular values equal to $1$. In \cref{sec:analysis}, we will explore the structure of the correlation matrix for non-Gaussian unitary operators.

Under the Jordan-Wigner transformation, Gaussian unitaries correspond to so-called \emph{matchgate circuits}, which were defined by Valiant as an example of quantum circuits that admitted efficient classical simulations \cite{valiant2001quantum}.
While matchgates circuits Gaussian unitaries are non-universal for quantum computation (and in fact efficiently simulable on classical computers), the addition of any non-Gaussian gate (most commonly the SWAP gate) forms a universal gate set that is capable of approximating arbitrary fermionic unitary operators. In general, we will say a fermionic unitary or gate is $\kappa$-local if it acts on at most $\kappa$ modes.

A particularly important state is the so-called \emph{fermionic EPR state} on $2n$ modes:
\[
\sigma = 2^{-2n}\prod_{a=1}^{2n} (\identity + i\majorana_{a} \majorana_{2n+a}).
\]

This state can be generated using Gaussian unitaries, making it a Gaussian state. Consider a unitary $U$ that acts only on the first $n$ modes (in particular, $U$ and $U^\dagger$ commutes with $\majorana_{2n+1},...,\majorana_{4n}$. Analogously to the qubit setting, we can define the fermionic Choi state $\sigma_U$ for $U$:
\[
\sigma_U = U\sigma U^\dagger = 2^{-2n} \prod_{a=1}^{2n} U(\identity + i\majorana_{a} \majorana_{2n+a})U^\dagger = 2^{-2n} \prod_{a=1}^{2n} (\identity + iU \majorana_{a} U^\dagger \majorana_{2n+a}).
\]
Thus, by the orthogonality of the $\majorana_a$, we have:
\[
\trace{-i\majorana_{a}\majorana_{2n + b} \cdot \sigma_U} = 2^{-n} \trace{\majorana_a U \majorana_b U^\dagger} = M_{ab}.
\]

Fermionic Choi states of $k$-mode unitaries are in exact correspondence with $2k$-mode fermionic states that are maximally entangled across the partition $([k], [2k]\setminus[k])$. 

Analogously to the qubit case, the fidelity between fermionic Choi states of two unitaries can be related to their Frobenius distance.

\begin{fact} \label{fact:choi-frobenius}
    Take $\sigma$ to be the fermionic EPR state:
    \[
    \sigma = 2^{-2n} \prod_{a=1}^{2n} \left(I + i \majorana_{a}\majorana_{2n+a}\right).
    \]
    Let $\sigma_U = (U \otimes I)\sigma(U \otimes I)^\dagger$ and let $\sigma_V$ be defined accordingly. Then $\frobeniusdistance{U,V} = \sqrt{1 - \trace{\sigma_U \sigma_V}}$.
\end{fact}

For a more thorough introduction to the Majorana representation of fermionic systems, we refer the reader to \cite{bravyi2004lagrangian}.

\section{The Structure of Gaussian Dimension} \label{sec:analysis}

We know that the correlation matrix of a Gaussian unitary is in $\orthgrp(2n)$ and provides a complete description of its action. What can the correlation matrix tell us about non-Gaussian unitaries? To answer this question, we formally define a notion of unitary \emph{Gaussian dimension}, which quantifies the degree of similarity of a given operator to a Gaussian one. Formally, the Gaussian dimension of a fermionic unitary $U$ on $n$ modes is defined as follows:
\begin{definition}\label{def:gaussian-dimension}
    We say a unitary $U$ on $n$ modes has Gaussian dimension $k$ if there exist operators $\newmajorana_1,...,\newmajorana_k$ such that:
    \begin{enumerate}[label=(\arabic*)]
        \item For each $a \in [k]$, $\newmajorana_a = \sum_{b=1}^{2n} M_{ba} \majorana_{b}$, where $\sum_{b=1}^{2n} M_{ba}^2 = 1$.
        \item For each $a,b \in [k]$ such that $a \neq b$, $\trace{\newmajorana_a \newmajorana_b} = 0$.
        \item For all $a \in [2n]$, $U \newmajorana_a U^\dagger = \sum_{b =1}^{2n} M_{ba}^\prime \majorana_b$, where $\sum_{b=1}^{2n}{M_{ba}^\prime}^2 = 1$.
    \end{enumerate} 
\end{definition}
Note that the Gaussian dimension is an integer between $0$ and $2n$ and is maximized for Gaussian unitaries, as one might expect.
At a high level, $U$ acts as a Gaussian operator on $k$ independent rotated Majorana operators, again conjugating them into a possibly different set of $k$ independent rotated Majorana operators. 
In this section, we prove a number of structural results about the Gaussian dimension that will be critical in analyzing our property testing algorithm in \cref{alg:property-testing} and our learning algorithm in \cref{sec:efficient-learning}.

\subsection{Gaussian Dimension and Correlation Matrices}

In this section, we prove a collection of basic facts about the Gaussian dimension and its connection to correlation matrices, motivating it as a measure of fermionic Gaussianity.
First, we show that the Gaussian dimension is invariant under multiplication by Gaussian unitaries.
\begin{proposition} \label{prop:gaussian-dim-invariance}
    Let $U$ be an $n$-mode fermionic unitary operator of Gaussian dimension at least $k$ and let $G_A$ and $G_B$ be arbitrary Gaussian operators. Then $G_A U$ and $U G_B$ both have Gaussian dimension at least $k$.
\end{proposition}
\begin{proof}
    By assumption, there exist Majorana bases $\{\newmajorana_1,...,\newmajorana_{2n}\}$ and $\{\newnewmajorana_1,...,\newnewmajorana_{2n}\}$ such that $U \newmajorana_a U^\dagger = \newnewmajorana_a$ for $a \in [k]$. It can easily be seen that every Gaussian unitary operator takes a rotated Majorana operator to some other rotated Majorana operator. Thus $G_A\newnewmajorana_a G_A^\dagger = \newnewmajorana_a^\prime$, where $\newnewmajorana_a^\prime$ is a rotated Majorana operator, for all $a \in [2n]$. As such, 
    \[
    G_A U \newmajorana_a U^\dagger G_A^\dagger = G_A \newnewmajorana_a G_A^\dagger = \newnewmajorana_a^\prime
    \]
    for all $a \in [k]$. This implies that $G_A U$ has Gaussian dimension at least $k$.

    Furthermore, we observe that for any unitary $U$ with Gaussian dimension at least $k$, its inverse $U^\dagger$ also has Gaussian dimension at least $k$ . Indeed, for $a \in [k]$ we have $U^\dagger \newnewmajorana_a U = \newmajorana_a$. Thus,  $U G_B$ has the same Gaussian dimension as $G_B^\dagger U^\dagger$. Since the inverse of any Gaussian operator is itself Gaussian, $G_B^\dagger U^\dagger$ has Gaussian dimension at least $k$ by the argument in the previous paragraph.
\end{proof}

Next, we connect the Gaussian dimension of a unitary to its circuit complexity by showing that the class of unitaries with high Gaussian dimension subsumes the class of unitaries containing a small number of local non-Gaussian gates. To do this, we will use a compression lemma of Mele and Herasymenko.

\begin{fact}[\cite{mele2024efficient}, Theorem 26] \label{fact:mele-herasymenko-compression}
    Let $U$ be fermionic unitary containing $t$ non-Gaussian gates that each act on at most $\kappa$ modes. Then there exist Gaussian unitaries $\leftgaussian, \rightgaussian$ and a unitary $u$ acting on at most $\kappa t$ modes such that
    \[
    U = G_A(\identity_{n - \kappa t} \otimes u) G_B^\dagger
    \]
\end{fact}

\begin{lemma}\label{lem:gates-gaussian-dim}
    Let $U$ consist of at most $t$ $\kappa$-local non-Gaussian gates. Then $U$ has Gaussian dimension at least $2n - 2\kappa t$.
\end{lemma}
\begin{proof}
    If $U$ is such an operator, then by \cref{fact:mele-herasymenko-compression}, there exist Gaussian unitaries $\leftgaussian$, $\rightgaussian$ and a unitary $u$ on $\kappa t$ fermionic modes such that
    \[
    U = \leftgaussian(\identity_{n - \kappa t} \otimes u) \rightgaussian^\dagger.
    \] 
    Suppose for each $a \in [2n]$, $\leftgaussian \majorana_a \leftgaussian^\dagger = \newnewmajorana_a$ and $\rightgaussian \majorana_a \rightgaussian^\dagger = \newmajorana_a$. The $\newmajorana_a$ and $\newnewmajorana_a$ are rotated majorana operators and the pairwise trace inner product of distinct $\newmajorana_a$ (and therefore of distinct $\newnewmajorana_a$) is easily seen to be $0$. For $a \leq 2(n - \kappa t)$, we have
    \begin{align*}
        U\majorana_aU^\dagger &= \leftgaussian(\identity_{n - \kappa t} \otimes u) \rightgaussian^\dagger \newmajorana_a \rightgaussian(\identity_{n - \kappa t} \otimes u)^\dagger \leftgaussian^\dagger \\
        &= \leftgaussian(\identity_{n - \kappa t} \otimes u) \majorana_a (\identity_{n - \kappa t} \otimes u)^\dagger \leftgaussian^\dagger \\
        &= \leftgaussian \majorana_a \leftgaussian^\dagger = \newnewmajorana_a.
    \end{align*}
   Thus $U$ has Gaussian dimension $2(n-\kappa t) = 2n - 2\kappa t$.
\end{proof}

From the proof above, it's clear that any unitary that can be \say{compressed} into the form $G_A(\identity_k \otimes u)G_B^\dagger$ has Gaussian dimension at least $2k$. As it turns out, this compressibility \emph{precisely} captures what it means to have high Gaussian dimension. Indeed, we show that all unitaries with high Gaussian dimension can be expressed in this form, generalizing \cref{fact:mele-herasymenko-compression}.

\begin{lemma}[Compression Lemma] \label{lem:compression}
    Let $U$ be a $n$-mode unitary of Gaussian dimension at least $k$. Then we can write $U = \leftgaussian (\identity_{\floor{k/2}} \otimes u)\rightgaussian^\dagger$, where $\leftgaussian$ and $\rightgaussian$ are Gaussian unitaries and $u$ acts nontrivially on at most $n - \floor{k/2}$ modes. 
\end{lemma}
\begin{proof}
    If $U$ has Gaussian dimension at least $k$ then there exist Majorana bases $\{\newmajorana_{1},...,\newmajorana_{2n}\}$ and $\{\newnewmajorana_{1},...,\newnewmajorana_{2n}\}$ such that $U \newmajorana_a U^\dagger = \newnewmajorana_a$ for $1 \leq a \leq k$. 
    We can define Gaussian operators $\leftgaussian$ and $\rightgaussian$ such that $\leftgaussian \majorana_a \leftgaussian^\dagger = \newnewmajorana_a$ and $\rightgaussian \majorana_a \rightgaussian^\dagger = \newmajorana_a$ for all $a \in [2n]$.
    
    Consider the unitary $W = \leftgaussian^\dagger U \rightgaussian$. It suffices to show that $W$ acts trivially on the first $\floor{k/2}$ modes. To do this, it suffices to show that $W$ commutes with $\majorana_a$ whenever $a \leq 2 \floor{k/2} \leq k$. In other words, demonstrating $W\majorana_a W^\dagger = \majorana_a$ for all such $a$ will complete the argument. Indeed, we have
    \begin{align*}
        W\majorana_a W^\dagger &= \leftgaussian^\dagger U \rightgaussian \majorana_a \rightgaussian^\dagger U^\dagger \leftgaussian \\
        &= \leftgaussian^\dagger U \newmajorana_a  U^\dagger \leftgaussian \\
        &= \leftgaussian^\dagger \newnewmajorana_a  \leftgaussian = \majorana_a. \qedhere
    \end{align*}
\end{proof}
Now, Gaussian unitaries are associated with orthogonal matrices, so the form of the compressed unitary seems conceptually similar to the singular value decomposition. Indeed, in the remainder of this section, we show that this intuition is spot on, and that the SVD of a correlation matrix reveals much about it's Gaussian structure.

As a first step, we show that the correlation matrix of the product of two unitaries is the product of their individual correlation matrices with the order reversed.

\begin{proposition}\label{prop:conj-mat-mult}
    Let $U_1$ and $U_2$ be unitaries with correlation matrices $M^{(1)}$ and $M^{(2)}$, respectively. Then the correlation matrix $M$ of $U = U_1U_2$ is given by $M^{(1)}M^{(2)}$.
\end{proposition}
\begin{proof}
    Fix $a,b \in [2n]$. We have 
    \begin{align*}
    M_{ab} &= 2^{-n}\trace{\majorana_a \cdot  U \majorana_b U^\dagger } \\
    &= 2^{-n}\trace{\majorana_a \cdot U_1 U_2 \majorana_b U_2^\dagger U_1^\dagger } \\
    &= 2^{-n} \sum_{c=1}^{2n} M^{(2)}_{cb}\trace{\majorana_a \cdot U_1 \majorana_c U_1^\dagger } \\
    &= 2^{-n} \sum_{c=1}^{2n}\sum_{d=1}^{2n} M^{(1)}_{ac} M^{(2)}_{cb}\trace{\majorana_a \majorana_d} \\
    &= \sum_{c=1}^{2n} M^{(1)}_{ac} M^{(2)}_{cb} = (M^{(1)}M^{(2)})_{ab}. \qedhere
    \end{align*}
\end{proof}

Next, we characterize the behavior of a unitary on a rotated Majorana operator given the correlation matrix of the unitary and the vector representation of the operator.

\begin{proposition}\label{prop:conjugate-mat-vec}
    Let $U$ and $V$ be unitary operators, where $V = \sum_{a = 1}^{2n} \alpha_a \majorana_a$, and let $W = UVU^\dagger$. Furthermore, let $U$ have correlation matrix $M$. Define vectors $v,w \in \R^{2n}$ by $v_a = \alpha_a$ and $w_a = 2^{-n} \trace{\majorana_a\cdot W}$. Then $w = Mv$.
\end{proposition}
\begin{proof}
    Fix any $a \in [2n]$. We have
    \begin{align*}
        w_a &= 2^{-n} \trace{UVU^\dagger \majorana_a} \\
         &= 2^{-n} \sum_{b=1}^{2n} v_{b} \trace{U\majorana_{b}U^\dagger \majorana_a} \\
         &= 2^{-n} \sum_{b=1}^{2n} \sum_{c=1}^{2n} M_{cb} v_{b} \trace{\majorana_{c} \majorana_a} \\
         &= \sum_{b=1}^{2n} M_{ab} v_{b} = Mv. \qedhere
    \end{align*}
\end{proof}

We are now ready to prove our main structural result about the correlation matrices of unitaries with Gaussian dimension $k$.

\begin{lemma}\label{lem:gaussian-dimension-svd} 
    Let $U$ be a unitary of Gaussian dimension $k$ with correlation matrix $M = V_A \Sigma V_B^\transpose$. Then $M$ has at least $k$ singular values equal to $1$ -- that is, $\sigma_1(M) = ... = \sigma_k(M)$. Furthermore, if $V_A,V_B \in \orthgrp(2n)$ correspond to Gaussian unitaries $\leftgaussian$ and $\rightgaussian$, respectively, then $\leftgaussian^\dagger U \rightgaussian = \identity_{\floor{k/2}} \otimes u$ for some unitary $u$ acting on $n-\floor{k/2}$ modes.
\end{lemma}
\begin{proof}
    By \Cref{lem:gates-gaussian-dim}, there exist rotated majorana operators $\newmajorana_{1},...,\newmajorana_{k}$ and $\newnewmajorana_{1 + 1},...,\newnewmajorana_{k}$ such that $U\newmajorana_{a}U^\dagger = \newnewmajorana_{a}$ for all $a \in [k]$. Let $v^{(1)},...,v^{(k)}$ and $w^{(1)},...,w^{(k)}$ be the coefficient vectors associated with the $\newmajorana_a$ and $\newnewmajorana_a$, respectively, as in \cref{prop:conjugate-mat-vec}. Then we have that $w^{(a)} = Mv^{(a)}$. Observe that $\langle v^{(a)}, v^{(b)} \rangle = 2^{-n} \trace{\eta_{a} \eta_{b}} = \delta_{ab}$ for all $a, b \in [k]$, and the analogous statement is easily seen to be true for the $w^{(a)}$ and $\newnewmajorana_a$. These facts together imply that at least $k$ singular values of $M$ are exactly equal to $1$, showing the first result.
    
    Now let $V_A$ correspond to $\rightgaussian$ and $V_B$ to $\leftgaussian$. We show that $\leftgaussian^\dagger \inputu  \rightgaussian = \identity_k \otimes \inneru$ for some $\inneru$ that acts only on the final $n - \floor{k/2}$ modes. Let us consider the correlation matrix $M^\prime$ of $U^\prime = \leftgaussian^\dagger \inputu  \rightgaussian$. By \cref{prop:conj-mat-mult}, $M^\prime = V_A^\transpose M V_B = \Sigma$. The first $k$ diagonal entries of $\Sigma$ are equal to $1$, which means that $U^\prime \majorana_a {U^\prime}^\dagger = \majorana_a$ for $a \in [k]$, implying that $U^\prime$ acts as the identity on the first $\floor{k/2}$ modes.
\end{proof}

\subsection{Robustness of Gaussian Dimension}

We show a pair of lemmas that together demonstrate that unitaries that are close to having high Gaussian dimension are close to unitaries that do have high Gaussian dimension. Essentially, if an operator has a correlation matrix whose singular values are sufficiently close to $1$, it can be \say{rounded} to a unitary whose correlation matrix corresponds to snapping the large singular values to $1$. By \cref{lem:gaussian-dimension-svd}, this new unitary has high Gaussian dimension.

First we show the result for Frobenius distance, which will be useful in showing our property testing result.

\begin{lemma} \label{lemma:rounded-frobenius-distance}
    Let $U$ be an $n$-mode unitary with correlation matrix $M = V_A \Sigma V_B^\transpose$, where $V_B$ and $V_A$ correspond to Gaussian unitaries $\leftgaussian$ and $\rightgaussian$, respectively. Furthermore, suppose $\sigma_1(M) = ... = \sigma_k(M) = 1$ and $\sigma_{k+ \ell}(M) \geq 1 - \tau$. Then there exists a unitary $U^\prime = \leftgaussian (\identity \otimes u^\prime) \rightgaussian^\dagger$ of Gaussian dimension $k + \ell$, where $u^\prime$ acts only on the last $n - \floor{(k + \ell)/2}$ modes, such that
    \[
    \frobeniusdistance{U, U^\prime} \leq 2\sqrt{\ell\tau}.
    \]
\end{lemma}
\begin{proof}
Take $W = \leftgaussian^\dagger U \rightgaussian$ and $t = n - \floor{k/2}$. Applying \cref{lem:compression}, we know that $W = \identity_{\floor{k/2}} \otimes \inneru$ for some unitary $\inneru$ acting on the final $t$ modes.
Define the $2n$-mode to $2(n-\floor{k/2})$-mode channel $\calC^{(k)}$ which traces out the system spanned by the Majorana operators $\majorana_1,...,\majorana_k$ and $\majorana_{2n+1},...,\majorana_{2n+k}$ and consider $\calC^{(k)}(\sigma_W)$. The latter is a pure state, since $\sigma_W$ is a product state across the corresponding cut. For $x \in \F_2^{\ell}$, define
\[
\rho_x = \prod_{a=1}^{\ell} \left(\frac{\identity_{2n} + (-1)^{x} i \majorana_{k + a}\majorana_{2n + k + a}}{2}\right)
\]

We can write
\[
\calC^{(k)}(\sigma_W) = \sum_{x \in \F_2^{\ell}} \alpha_x \rho_x \otimes \sigma_W^{(x)}
\] 
where each $\sigma_W^{(x)}$ is a density matrix that can be represented using only the Majorana operators $\majorana_{k+\ell+1},...,\majorana_{2n}$ and $\majorana_{2n+k+\ell+1},...,\majorana_{4n}$.
Now, since $\calC^{(k)}(\sigma_W)$ is maximally entangled across the central cut, each $\sigma_W^{(x)}$ is as well, and thus corresponds to a Choi state of a unitary acting on the final $n - \floor{(k + \ell)/2}$ modes. Define this unitary to be $u^\prime$. We have that
\[
\frobeniusdistance{u, \identity_{\floor{(k+\ell)/2} - \floor{k/2}} \otimes u^\prime} = \sqrt{1 - \trace{\calC^{(k)}(\sigma_W) \cdot \rho_{0^\ell} \otimes \sigma_W^{(0^\ell)}}} = \sqrt{1 - \alpha_{0^\ell}}.
\]
If $k + \ell$ is odd, then we have further that $\innerroundedu \majorana_{k + \ell} = \majorana_{k + \ell}$. Thus, applying \cref{prop:gaussian-dim-invariance}, $U^\prime = G_A (\identity_{(k + \ell)/2} \otimes \innerroundedu)G_B^\dagger$ has Gaussian dimension at least $k+ \ell$. To show that $\frobeniusdistance{U, U^\prime} \leq 2 \sqrt{\ell \tau}$, it suffices to show that $\sqrt{1 - \alpha_{0^\ell}} \leq 2\sqrt{\ell \tau}$.
We will do just this, using the quantum union bound to show that $\alpha_{0^\ell} \geq 1 - 4 \ell \tau$, and completing the proof. Indeed, observe that for each $a \in [2n]$, we have
\[
\trace{\left(\frac{\identity_{2n} + i \majorana_a \majorana_{2n + a}}{2}\right) \sigma_W} = \sigma_a(M).
\]
Thus for $k+1 \leq a \leq k + \ell$, we have
\[
\trace{\left(\frac{\identity_{2n} + i \majorana_a \majorana_{2n + a}}{2}\right) \calC^{(k)}(\sigma_W)} \geq 1 - \tau
\]
by assumption. The quantity $\alpha_0$ is exactly the probability that the projective measurements corresponding to $\left(\frac{\identity_{2n} + i \majorana_a \majorana_{2n + a}}{2}\right)$ are simultaneously satisfied for all $a \in \{k+1,...,k+\ell\}$. By the quantum union bound (\cref{fact:quantum-union-bound}), this happens with probability at least $1 - 4 \ell \tau$. Thus $\alpha_{0^\ell} \geq 1 - 4 \ell \tau$ and the claim is shown.
\end{proof}

Next, we build upon this result to show when the singular values are even closer to $1$, the rounded unitary is close to a unitary of higher Gaussian dimension, even in diamond distance. This lemma will be central in showing that our learning algorithm recovers the input unitary to small diamond error.

\begin{lemma}\label{lem:rounded-diamond-distance}
    Let $U$ be an $n$-mode unitary with Gaussian dimension at least $2(n-t)$ and correlation matrix $M = V_A \Sigma V_B^\transpose$. Let $V_A$ and $V_B$ correspond to Gaussian unitaries $\leftgaussian$ and $\rightgaussian$, respectively. Furthermore, suppose 
    \[
    \sigma_{2n-2t^\prime}(M) \geq 1 - \frac{\eps^2}{16t \cdot 2^t}
    \]
    for some $0 \leq t^\prime \leq t$.
    Then there exists a unitary $U^\prime = \leftgaussian (\identity \otimes u^\prime) \rightgaussian^\dagger$, where $u^\prime$ acts only on the last $\pt$ modes, such that
    \[
    \diamonddistance{U}{U^\prime} \leq \eps.
    \]
\end{lemma}
\begin{proof}
    Denote $k = n - t$ and $\pk = n - \pt$.
    Applying \cref{lemma:rounded-frobenius-distance}, we know that there exists a $U^\prime = \leftgaussian (\identity \otimes u^\prime) \rightgaussian^\dagger$ such that $u^\prime$ acts only on the final $t^\prime$ modes such that
    \[
    \frobeniusdistance{U, U^\prime} \leq 2\sqrt{\frac{2(t-t^\prime)\eps^2}{16 t\cdot 2^t}} \leq \frac{\eps}{\sqrt{2} \cdot 2^{t/2}}.
    \]
    We can write $U = \leftgaussian (\identity_{k} \otimes u) \rightgaussian^\dagger$ where $u$ acts only on the last $t$ modes. Now, by the unitary invariance of diamond norm, we have
    \[
    \diamonddistance{U}{U^\prime} = \diamonddistance{\leftgaussian^\dagger U \rightgaussian}{\leftgaussian^\dagger U^\prime \rightgaussian} = \diamonddistance{\identity_{k} \otimes u}{\identity_{k} \otimes \identity_{\pk - k} \otimes u^\prime}
    \]
    where in the term $\identity \otimes \identity \otimes u^\prime$, the first identity acts on the first $k$ modes and the second identity acts on the next $t - t^\prime$ modes. Since the action of $\identity \otimes u$ and $\identity \otimes \identity \otimes u^\prime$ are identical on the first $k$ modes, we have
    \begin{align*}
    \diamonddistance{\identity_k \otimes u}{\identity_k \otimes \identity_{\pk - k} \otimes u^\prime} &= \diamonddistance{u}{\identity_{\pk - k} \otimes u^\prime} \\
    &\leq \sqrt{2} \cdot 2^{t/2}\cdot \frobeniusdistance{u, \identity_{\pk - k} \otimes u^\prime} \\
    &= \sqrt{2} \cdot 2^{t/2}\cdot \frobeniusdistance{\identity_k \otimes u, \identity_k \otimes \identity_{\pk - k} \otimes u^\prime} \\
    &= \sqrt{2} \cdot 2^{t/2}\cdot \frobeniusdistance{\leftgaussian(\identity_k \otimes u)\rightgaussian^\dagger, \leftgaussian(\identity_k \otimes \identity_{\pk - k} \otimes u^\prime)\rightgaussian^\dagger} \\ 
    &= \sqrt{2} \cdot 2^{t/2}\cdot \frobeniusdistance{U,U^\prime} \leq \eps. \qedhere
    \end{align*}
\end{proof}

\Cref{lemma:rounded-frobenius-distance} (and thus \cref{lem:rounded-diamond-distance}) works even in the case that the singular value decomposition of $M$ is not properly sorted. In particular, it holds as long as $\sigma_{a}(M) \geq 1 - \tau$ for all $a \in [\pk]$ and for $k$ values of $a \in [\pk]$, $\sigma_{a}(M) = 1$. In this case, $u$ acts on the final $n - \floor{\pk/2}$ modes as well as a size $\floor{\pk/2} - \floor{k/2}$ subset of the first $\pk$ modes. That $\inneru$ simply acts on the final $n-\floor{k/2}$ modes was assumed to simplify the argument. This nuance will be relevant in the proof of the main result, \cref{thm:main-learning}.

\subsection{Gaussian Operators Induced by Correlation Matrices}
Suppose we take the singular value decomposition of two correlation matrices that are close: does this mean that the Gaussian operators induced by their left-singular and right-singular vectors are close? In this section, we will argue that this is indeed the case granted that some assumptions are met.

First we relate the diamond distance of two Gaussian operators to the angle distance of the columns of their correlation matrices.

\begin{proposition}\label{lem:angle-diamond-distance}
    Let $G$ and $\wh{G}$ be two $n$-mode Gaussian operators that correspond to matrices $M, \wh{M} \in \orthgrp(2n)$, respectively. 
    Then
    \[
    \diamonddistance{G}{\wh{G}} \leq 4n^{3/2}\cdot  \sin \theta (M, \wh{M}).
    \]
\end{proposition}
\begin{proof}
    First, we show that $M$ and $\wh{M}$ have small operator distance. Indeed, taking $u_a$ and $\wh{u_a}$ to be the columns of $M$ and $\wh{M}$, respectively, we have
    \begin{align*}
        \norm{M - \wh{M}} &\leq \frobeniusnorm{M - \wh{M}} \\
        &= \sqrt{2 \sum_{a=1}^{2n} (1 - \langle u_a, \wh{u_a} \rangle)} \\
        &\leq \sqrt{2 \sum_{a=1}^{2n} (1 - \cos \theta (M, \wh{M}))} \\
        &\leq \sqrt{2 \sum_{a=1}^{2n} (1 - \cos^2 \theta (M, \wh{M}))} \\
        &= \sqrt{2 \sum_{a=1}^{2n} \sin^2 \theta (M, \wh{M})} =  2\sqrt{n} \cdot \sin \theta (M, \wh{M}). 
    \end{align*}
    Applying \cref{fact:gaussian-diamond-distance}, we have
    \[
    \diamonddistance{G}{\wh{G}} \leq 2n \norm{M - \wh{M}} \leq 4n^{3/2}\cdot \sin \theta (M, \wh{M}).  \qedhere.
    \]
\end{proof}

Now we show that if $M \in \R^{2n \times 2n}$ are correlation matrices that is well-approximated by $\wh{M} \in \R^{2n \times 2n}$ and the two matrices satisfy an eigengap condition, a la \cref{fact:davis-kahan}, then one can find Gaussian operators associated with their singular value decompositions that are correspondingly close in diamond distance.

\begin{lemma}\label{lem:svd-gaussian-error-bound}
    Let $M$ be the correlation matrix of a $n$-mode unitary operator and consider $\wh{M} \in \R^{2n \times 2n}$ such that:
    \begin{enumerate}[label=(\arabic*)]
        \item $\wh{M}$ has singular value decomposition $\wh{V_A}\wh{\Sigma}\wh{V_B}^\transpose$, where $\wh{V_A}$ and $\wh{V_B}$ are associated with Gaussian unitaries $\estleftgaussian$ and $\estrightgaussian$, respectively.
       
        \item For some $k \in [2n]$, $\sigma_k(M) - \sigma_{k+1}(\wh{M}) \geq \eigengap$ and $\sigma_{k}(M) \geq 2/3$ (where $\sigma_{2n+1}(\wh{M})$ is taken to be $-\infty)$.

         \item $\norm{\wh{M}} \leq K$ and $\norm{M - \wh{M}} \leq \esterror$.
    \end{enumerate}
    Then there exists a singular value decomposition of $M = V_A \Sigma V_B^\transpose$, where $V_A$ and $V_B$ are associated with respective Gaussian unitaries $\leftgaussian$ and $\rightgaussian$ such that 
    \[
    \diamonddistance{\leftgaussian}{\estleftgaussian} \leq \frac{12K \esterror \cdot  n^{3/2}}{\eigengap}, \quad \diamonddistance{\rightgaussian}{\estrightgaussian} \leq \frac{12K \esterror \cdot  n^{3/2}}{\eigengap}.
    \]
\end{lemma}
\begin{proof}
    Consider the symmetric matrices $R = MM^\transpose$ and $\wh{R} = \wh{M}\wh{M}^\transpose$.
    Let $\calE$ be the subspace spanned by the first $k$ eigenvectors of $R$ and $\wh{R}$, respectively. Observe that $\wh{\calE}$ is also the subspace spanned by the first $k$ columns of $\wh{V_A}$. We have that 
    \begin{align*}
    \lambda_k(M) - \lambda_{k+1}(\wh{M}) &= \sigma_k(M)^2 - \sigma_{k+1}(\wh{M})^2 \\
    &= \left(\sigma_k(M) + \sigma_{k+1}(\wh{M})\right)\left(\sigma_k(M) - \sigma_{k+1}(\wh{M})\right) \\
    &\geq \frac{2}{3}\left(\sigma_k(M) - \sigma_{k+1}(\wh{M})\right) \geq \frac{2\eigengap}{3}
    \end{align*}
    Furthermore, we have that
    \begin{align*}
        \norm{R - \wh{R}} &= \norm{MM^\transpose - \wh{M}\wh{M}^\transpose} \\
        &= \norm{MM^\transpose - M\wh{M}^\transpose + M\wh{M}^\transpose - \wh{M}\wh{M}^\transpose} \\
        &\leq \norm{MM^\transpose - M\wh{M}^\transpose} + \norm{M\wh{M}^\transpose - \wh{M}\wh{M}^\transpose} \\
        &\leq \norm{M}\norm{M^\transpose - \wh{M}^\transpose} + \norm{\wh{M}^\transpose}\norm{M - \wh{M}} \leq 2K\esterror.
    \end{align*}
    Applying \cref{fact:davis-kahan}, we have that 
    \[
    \sin \theta (\calE, \wh{\calE}) \leq \frac{\norm{R - \wh{R}}}{\lambda_k(R) - \lambda_{k+1}(\wh{R})} \leq \frac{3K \esterror}{\eigengap},
    \]
    We can thus construct a basis for $\calE$ whose angles with corresponding columns of $\wh{V_A}$ is at most $\arcsin(3K\esterror/\eigengap)$. Take the first $k$ columns of $V_A$ to be this basis.
    
    To construct the final $n-k$ columns of $V_A$ we consider the matrices $-R$ and $-\wh{R}$. Indeed, just like as above,
    \[
    \lambda_{n-k-1}(-\wh{R}) - \lambda_{n-k}(-R) \geq \frac{2\eigengap}{3}
    \]
    and obviously $-R$ and $-\wh{R}$ have the same operator distance as $R$ and $\wh{R}$. 
    As such, taking $\calF$ and $\wh{\calF}$ to be the respective subspaces spanned by the first $2n-k$ principal eigenvectors we conclude similarly that
    \[
    \sin \theta (\calF, \wh{\calF}) \leq \frac{3K \esterror}{\eigengap}.
    \]
    The last $2n-k$ columns of $\wh{V_A}$ are exactly the first $2n-k$ principal eigenvectors of $-\wh{R}$. This means we can similarly construct a basis for $\calF$ whose angles with corresponding columns of $\wh{V_A}$ is at most $\arcsin(3K\esterror/\eigengap)$ and take the elements of this basis to be the last $2n-k$ columns of $V_A$. As such, we have $\sin \theta(V_A, \wh{V_A}) \leq \frac{3K\esterror}{\eigengap}$. By \cref{lem:angle-diamond-distance}, this means that the associated Gaussian unitaries $\leftgaussian$ and $\estleftgaussian$ satisfy
    \[
    \diamonddistance{\leftgaussian}{\estleftgaussian} \leq 4n^{3/2} \sin \theta(V_A, \wh{V_A}) \leq \frac{3K\esterror}{\eigengap} \leq \frac{12 K \esterror \cdot n^{3/2}}{\eigengap}.
    \]
    
    Thus the statement holds for the $\leftgaussian$ and $\estleftgaussian$.
    The analogous statement for $\rightgaussian$ and $\estrightgaussian$ follows simply by considering $M^\transpose M$ and $\wh{M}^\transpose \wh{M}$ instead of $MM^\transpose$ and $\wh{M}\wh{M}^\transpose$ and repeating the argument above.
\end{proof}

We comment on a slight nuance: the singular value decomposition produced for $M$ may not put the singular values in decreasing order. However, they satisfy the condition that the first $k$ singular values are all larger than the next $2n-k$, which is the detail that matters in applying \cref{lemma:rounded-frobenius-distance,lem:rounded-diamond-distance}.

\section{Property Testing Unitary Gaussian Dimension}\label{sec:property-testing}

Our first algorithmic result is a property tester for unitaries of Gaussian dimension at least $k$. 
The algorithm functions by estimating the $k$-th largest singular value of the correlation matrix. By \cref{lem:gaussian-dimension-svd}, this value is exactly $1$ for unitaries of Gaussian dimension at least $k$. On the other hand, whenever this value is sufficiently large, the input unitary can be rounded to a Gaussian dimension $k$ unitary without much loss in Frobenius distance via \cref{lemma:rounded-frobenius-distance}. As such, estimating this singular value to sufficiently small distance will allow us to accurately distinguish the two cases.

\begin{algorithm}[ht]
    \caption{Property Testing Unitaries of Gaussian Dimension $k$} \label{alg:property-testing}
    \KwInput{Oracle access to a fermionic unitary $U$ on $n$ modes, along with an integer $k \in [2n]$.}
    \KwPromise{$U$ is either a unitary of Gaussian dimension at least $k$ or $\eps$-far in Frobenius distance from all such unitaries.}
    \KwOutput{\texttt{YES} if $U$ is a unitary of Gaussian dimension at least $k$ or \texttt{NO} otherwise.}

    Set $\mtest = 243\cdot nk^2\log(2n/\delta)/\eps^4$.

    Initialize $\wh{M} = 0$.

    \For{$a \in [2n]$}{
        \For{$b \in [2n]$} {
            \RepTimes{$\mtest$} {
                Prepare the state $\sigma_U = (U \otimes I) \left(\prod_{a=1}^{2n} \frac{I + i \majorana_{a} \majorana{a+2n}}{2} \right)(U^\dagger \otimes I)$.

                Let $X$ be the result of measuring $\sigma_U$ in the eigenbasis of the operator $-i \majorana_{b} \majorana_{a+2n}$.

                Increment $\wh{M}_{ab}$ by $X/\mest$.
            }
        }
    }

    Compute the top $k$ singular values of $\wh{M}$: $\sigma_1(\wh{M}) \geq ... \geq \sigma_k(\wh{M})$.

    \textbf{return} \texttt{YES} if $\sigma_k(\wh{M}) \geq 1 - \frac{\eps^2}{8k}$ and \texttt{NO} otherwise.
    
\end{algorithm}

\begin{theorem}
    Let $U$ be a $n$-mode unitary that is either:
    \begin{enumerate}[label=(\arabic*)]
        \item Of unitary Gaussian dimension at least $k$.
        \item $\eps$-far in Frobenius distance from all such unitaries.
    \end{enumerate}
    With probability at least $1-\delta$, \cref{alg:property-testing} successfully decides which is the case. Furthermore, it runs in time $\poly(n, k, 1/\eps, \log(1/\delta)).$
\end{theorem}

\begin{proof}
    Let $M$ be the correlation matrix of $U$. Applying \cref{fact:matrix-bernstein}, the probability that $\norm{\wh{M} - M} > \frac{\eps}{9k}$ is at most
    \[
    2n\exp\left(-\frac{\mtest\eps^4}{243nk^2}\right) \leq \delta
    \]
    So, with probability at least $1-\delta$, $\norm{\wh{M} - M} \leq \frac{\eps^2}{9k}$. We will argue that whenever this holds, the algorithm is correct, and the statement will follow.
    
    In the first case, $M$ is indeed a unitary of Gaussian dimension at least $k$. By \cref{lem:gaussian-dimension-svd}, this means that $\sigma_k(M) = 1$. Furthermore, by \cref{prop:singular-value-closeness}, $\abs{\sigma_k(M) - \sigma_k(\wh{M})} \leq \frac{\eps^2}{9k}$, so 
    \[
    \sigma_k(\wh{M}) \geq 1 - \frac{\eps^2}{9k} > 1 - \frac{\eps^2}{8k}.
    \]
    Therefore, in this case, the algorithm always accepts.

    Now assume that the algorithm accepts with probability at least $1-\delta$. This happens whenever $\sigma_k(\hat{M}) \geq 1- \frac{\eps^2}{8k}$. By our assumption on the operator distance between $M$ and $\hat{M}$, we have
    \[
    \sigma_k(M) \geq 1 - \frac{\eps^2}{8k} - \frac{\eps^2}{9k} \geq 1 - \frac{\eps^2}{4k}.
    \]
    
    By \cref{lemma:rounded-frobenius-distance}, this implies that $U$ has Frobenius distance at most $\eps$ to some unitary $U^\prime$ of Gaussian dimension at least $k$. 
\end{proof}

\section{Learning Algorithm} \label{sec:efficient-learning}

In this section, we present our main learning algorithm and prove its correctness. Given a $n$-mode unitary $U$, promised to be of Gaussian dimension at least $k$, the algorithm outputs a circuit that implements a unitary $\wt U$ with Gaussian dimension at least $2 \floor{k/2}$ such that $\diamonddistance{U}{\wt U} \leq \eps$ with high probability.

\begin{algorithm}[H]
    \caption{Learning Unitaries of High Gaussian Dimension}\label{alg:learning}
    \KwInput{A fermionic unitary $U$ on $n$ modes, along with an integer $k \in \{0,...,2n\}$.}
    \KwPromise{$U$ has Gaussian dimension at least $k$.}
    \KwOutput{A circuit implementing a unitary $\wt{U}$ such that $\diamonddistance{\wt{U}}{U} \leq \eps$.}

    Define $t = 2(n - \floor{k/2})$.
    
    Let $\alpha = \frac{\eps^{3}}{C \cdot n^{3/2}\cdot t(t+1) \cdot 2^{t/2}}$ for some suitably large constant $C$ and take $\mest = 3n\log(4n/\delta)/\alpha^2$.

    Initialize a $2n \times 2n$ matrix $\wh{M} = 0$.

    \For{$a \in [2n]$}{
        \For{$b \in [2n]$} {
            \RepTimes{$\mest$} {
                Prepare the state $\sigma_U = (U \otimes I) \left(\prod_{a=1}^{2n} \frac{I + i \majorana_{a} \majorana_{a+2n}}{2} \right)(U^\dagger \otimes I)$.

                Let $X$ be the result of measuring $\sigma_U$ in the eigenbasis of the operator $-i \majorana_{b} \majorana_{a+2n}$.

                Increment $\wh{M}_{ab}$ by $X/\mest$.
            }
        }
    }

    Take the singular value decomposition of $\wh{M} = \wh{V_{A}} \wh{\Sigma} \wh{V_{B}}^\transpose$. \label{line:svd}

    Construct Gaussian unitaries $\estleftgaussian$ and $\estrightgaussian$ corresponding to $\wh{V_A}$ and $\wh{V_B}$, respectively. \label{line:gaussian-construct}

    Define $\pk$ to be the largest $\ell \in \{k,...,2n\}$ such that $\sigma_\ell(\wh{M}) \geq 1 - \frac{(\ell - k + 1)\eps^2}{700\cdot t(t + 1)\cdot 2^{t/2}}$. \label{line:sv-partitioning}

    Perform $(\eps/9, \delta/2)$-tomography on the last $n - \floor{\pk/2}$ modes of $\estleftgaussian^\dagger U \estrightgaussian$ to obtain a channel $\innertomographyu$. \label{line:tomography}

    Round $\innertomographyu$ to a unitary operator $\innerpolaru$ using the polar decomposition. \label{line:unitary-rounding}

    Construct a circuit that implements a unitary $\innercircuitu$ that approximates $\innerpolaru$ to error $\eps/9$. \label{line:circuit synthesis}

    \textbf{return:} A circuit for $\estrightgaussian^\dagger$, followed by a circuit for $\identity_{\floor{\pk/2}} \otimes \innercircuitu$, and finally a circuit for $\estleftgaussian$.
\end{algorithm}
\vspace{1em}

The steps in \cref{alg:learning} follow closely from the overview provided in \cref{sec:technical-overview}. 
We discuss the singular value partitioning step that occurs in \cref{line:sv-partitioning} in more depth. If the input unitary $U$ is promised to be of Gaussian dimension at least $k$, then at least $k$ singular values of $M$ are promised to be $1$. So, we can assume the first $k$ singular values of $\wh{M}$ are effectively $1$. We now must consider the remaining $2n-k$ singular values, indexed by $\ell \in \{k+1,...,2n\}$. The threshold in \cref{line:sv-partitioning} decreases for every value of $\ell$, and this allows us to induce an eigengap. A diagram is provided in \Cref{fig:partitioning}

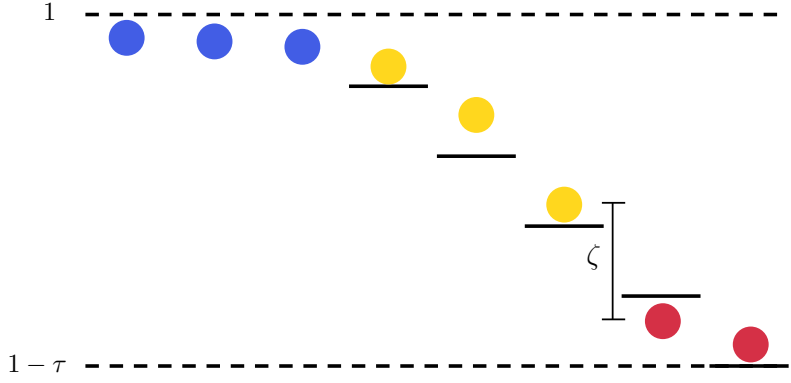
\begin{figure}
    \centering
    \scalebox{0.9}{\input{diagrams/rounding.tikz}}
    \caption{A depiction of the singular value partitioning step in \cref{line:sv-partitioning} of \cref{alg:learning} for $n = 4$ and $k = 3$. The circles from left to right represent $\sigma_1(\wh{M}),...,\sigma_8(\wh{M})$. $U$ is promised to be of Gaussian dimension at least $3$, so the first $3$ singular values are safe and marked \textbf{\color{figBlue}blue}. Then, for the next $5$ singular values, we apply successively lower thresholds (represented by \textbf{solid} bars) until a failure occurs. Singular values that pass the check are considered ``good'' and marked \textbf{\color{figYellowText}yellow}, and singular values that fail are considered ``bad'' and marked \textbf{\color{figRed}red}.
    The $7$th singular value is bad, so every following singular value is bad as well. The successively lower thresholds guarantee that the gap between the last good and first bad singular values is bounded below by some $\eigengap$. 
    The lowest threshold $1 - \tau$ is chosen so that if all singular values are good, $U$ can be rounded to a fully Gaussian unitary without much loss in diamond distance.
    }
    \label{fig:partitioning}
\end{figure}

We are now ready to prove the correctness of \cref{alg:learning}.

\begin{theorem}\label{thm:main-learning}
    Let $U$ be an $n$-mode fermionic unitary operator with Gaussian dimension at least $k$ and define $t = 2(n - \floor{k/2})$.
    With probability at least $1-\delta$, \cref{alg:learning} returns a circuit that implements a unitary $\wt{U}$ such that $\diamonddistance{U}{\wt{U}} \leq \eps$. Furthermore, it runs in time $\poly(n, 2^t, 1/\eps, \log(1/\delta))$.
\end{theorem}

\begin{proof}
    Throughout the proof, we will assume that the algorithm succeeds in estimating $M$ to a sufficient operator distance. In particular, 
    defining
    \[
    \alpha = \frac{\eps^{3}}{180000 \cdot n^{3/2}\cdot t(t+1) \cdot 2^{t/2}}, 
    \]
    we assume that $\norm{M - \wh{M}} \leq \alpha$, which happens except with probability at most
    \[
    2n \exp\left(- \frac{\mest \cdot \alpha^2}{3n}\right) \leq \delta/2
    \]
    according to \cref{fact:matrix-bernstein}. By \cref{prop:singular-value-closeness}, a consequence of this fact is that 
    \[
    \abs{\sigma_a(M) - \sigma_a(\wh{M})} \leq \alpha
    \]
    for all $1 \leq a \leq 2m$. 
    
    Define $\pt = 2(n - \floor{\pk/2})$ and note that $\pk \geq k$ and $\pt \leq t$.
    Since for $n$ sufficiently large we have
    \[
     \sigma_{\pk}(M) \geq 1 - \frac{(2n - k + 1)\eps^2}{700\cdot t(t+1) \cdot 2^{t/2}} - \alpha \geq 1 - \frac{\eps^2}{648\cdot t \cdot 2^{t/2}},
    \]
    $U$ is $\eps/9$-close in diamond distance to a unitary $U^\prime$ of Gaussian dimension $k^\prime$ by \cref{lem:rounded-diamond-distance}. Furthermore, for any Gaussian unitaries $\leftgaussian,\rightgaussian$ such that $\leftgaussian^\dagger U \rightgaussian = \identity_{\floor{k/2}} \otimes \inneru$ (where $\inneru$ acts only on the final $\ceil{t/2}$ modes), there exists a unitary $\innerroundedu$ acting on the final $\ceil{t^\prime/2}$ modes such that $U^\prime = \leftgaussian^\dagger (\identity_{\floor{\pk/2}} \otimes \innerroundedu) \rightgaussian$. 
    Such $\leftgaussian$ and $\rightgaussian$ are guaranteed to exist by \cref{lem:compression,lem:gaussian-dimension-svd}.
    We will argue that the unitary $\wt{U}$ constructed by the algorithm is $8\eps/9$-close to such a $U^\prime$, which yields our main result by a simple application of the triangle inequality.

    The next step of the algorithm is to perform the singular value decomposition on $\wh{M} = \wh{V_A} \wh{\Sigma} \wh{V_B}^\transpose$. 
    We can associate $\wh{V_A}$ and $\wh{V_B}$ with the Gaussian operators $\estrightgaussian$ and $\estleftgaussian$ respectively, a la \cref{lem:gaussian-dimension-svd}.
    We will argue that $\estleftgaussian$ and $\estrightgaussian$ are respectively close to some $\leftgaussian$ and $\rightgaussian$ that diagonalize $\inputu$ (and therefore $\roundedu$). 

    Indeed, the rounding step in \cref{line:sv-partitioning}
    ensures that
    \[
    \sigma_{\pk}(M) - \sigma_{\pk + 1}(\wh{M}) \geq \frac{\eps^2}{700\cdot t(t+1) \cdot 2^{t/2}} - \alpha \geq \frac{\eps^2}{750\cdot t(t+1) \cdot 2^{t/2}} = \eigengap,
    \]
    where the second inequality holds when $n$ is sufficiently large.
    Furthermore, we have $\norm{\wh{M}} \leq \norm{M} + \alpha$: we will bound this by $10/9$ for convenience. 
    Applying \cref{lem:svd-gaussian-error-bound}, there exists a singular value decomposition $M = V_A \Sigma V_B^\transpose$ where $V_B$ and $V_A$ are respectively associated with Gaussian operators $\leftgaussian$ and $\rightgaussian$ such that
    \[
    \diamonddistance{\leftgaussian}{\estleftgaussian} \leq \frac{12 K \alpha \cdot n^{3/2}}{\eigengap} \leq 
    \frac{\eps}{18}, \quad \diamonddistance{\rightgaussian}{\estrightgaussian} \leq \frac{12 K \alpha \cdot n^{3/2}}{\eigengap} \leq 
    \frac{\eps}{18}
    \]

    The final phase of the algorithm begins by performing full unitary tomography on $\estleftgaussian^\dagger U \estrightgaussian$ after tracing out the first $\floor{\pk/2}$ modes.
    Let $\calC^{(k)}$ be the channel that corresponds to tracing out the first $\floor{k/2}$ modes. In other words, $\calC^{(k)}(\identity_{\floor{k/2}} \otimes u) = u$. Since $\calC^{(k)}$ is a completely-positive trace-preserving map, by \cref{fact:diamond-distance-contraction} we have that for any quantum channels $\calU, \calV$:
    \[
    \diamonddistance{\calC^{(k)}\calU}{  \calC^{(k)}\calV} \leq \diamonddistance{\calU}{\calV}.
    \]
     So, \cref{line:tomography} of the algorithm performs full process tomography on $\calC^{(\pk)}(\estleftgaussian^{\dagger}U\estrightgaussian)$ to obtain a channel $\innertomographyu$ on $\ceil{\pt/2}$ modes. This tomography step has diamond error $\eps/9$ and failure probability $\delta/2$. 
    Now, since $\diamonddistance{\identity_{\floor{\pk/2}} \otimes u^\prime}{\estleftgaussian^{\dagger}U\estrightgaussian} \leq 2\eps/9$, we know that 
    \[
    \diamonddistance{\innerroundedu}{\innertomographyu} \leq \diamonddistance{\innerroundedu}{\calC^{(\pk)}(\estleftgaussian^{\dagger}U\estrightgaussian)} + \diamonddistance{\calC^{(\pk)}(\estleftgaussian^{\dagger}U\estrightgaussian)}{\innertomographyu} \leq \frac{2\eps}{9} + \frac{\eps}{9} = \frac{\eps}{3}.
    \]
    Thus, $\innertomographyu$ is $\eps/3$-far in diamond distance from some unitary operator. Performing the unitary rounding step in \cref{line:unitary-rounding} results in a unitary operator $\innerpolaru$ that is at most $\eps/3$-far in diamond distance from $\innertomographyu$, since the polar decomposition finds the closest unitary operator in operator distance. Then the circuit synthesis step in \cref{line:circuit synthesis} produces a circuit for a unitary $\innercircuitu$ that has diamond distance at most $\eps/9$ from $\innerpolaru$. All-in-all, \cref{line:tomography} to \cref{line:circuit synthesis} of the algorithm produce a unitary $\innercircuitu$ such that $\diamonddistance{\innercircuitu}{\innerroundedu} \leq 7\eps/9$.

    The circuit returned by our algorithm implements the unitary $\outputu = \estleftgaussian (\identity_{\floor{\pk/2}} \otimes \innercircuitu) \estrightgaussian$. Recall that $\diamonddistance{\estleftgaussian}{\leftgaussian} \leq \eps/18$ and $\diamonddistance{\estrightgaussian}{\rightgaussian} \leq \eps/18$. Putting everything together, we have 
    \begin{align*}
    \diamonddistance{\outputu}{\inputu} &\leq \diamonddistance{\outputu}{\roundedu} + \diamonddistance{\roundedu}{\inputu}
    \\
    &\leq \diamonddistance{\estleftgaussian (\identity_{\floor{\pk/2}} \otimes \innercircuitu) \estrightgaussian}{\leftgaussian (\identity_{\floor{\pk/2}} \otimes \innerroundedu) \rightgaussian} + \frac{\eps}{9} \\
    &\leq \diamonddistance{\leftgaussian (\identity_{\floor{\pk/2}} \otimes \innercircuitu) \rightgaussian}{\leftgaussian (\identity_{\floor{\pk/2}} \otimes \innerroundedu) \rightgaussian} + \frac{\eps}{9} + \frac{\eps}{9} \\
    &= \diamonddistance{\innercircuitu}{\innerroundedu} + \frac{\eps}{9} + \frac{\eps}{9} \leq \eps.
    \end{align*}

    Thus, under the assumptions that the estimate $\hat{M}$ of $M$ is sufficiently accurate and that the tomography of $\innertomographyu$ succeeds, our output circuit has diamond distance at most $\eps$ from the input $U$. By the union bound, this happens except with probability at most $\delta/2 + \delta/2 = \delta$.

    We now analyze the runtime of the algorithm. Estimating $\wh{M}$ takes $4n^2 \mest$ time and samples, which is clearly $\poly(n,2^t,1/\eps)$. Performing the singular value decomposition and finding the cutoff singular value takes time $\poly(n)$. The full unitary tomography step takes $\poly(2^{\pt}, 1/\eps)$ time and samples, and since $\pt \leq t$, this is at most $\poly(2^{t}, 1/\eps)$. Finally, the rounding and unitary synthesis steps take $\poly(2^{\pt}, 1/\eps)$ time. Overall, the time and sample complexity of the algorithm are dominated by the estimation of $\wh{M}$, and clearly has time complexity $\poly(n,2^t,1/\eps)$. \qedhere
\end{proof}

We realize that the hybrid argument in the above proof, which bounds the distance between $\inputu$ and $\outputu$ based on the estimates $\estleftgaussian,\estrightgaussian,$ and $\innercircuitu$, might be hard to follow. As such, we provide an outline in \Cref{fig:hybrid-argument}.

We conclude by discussing potential improvements to \cref{alg:learning}. First, we address the question of whether it is possible to \say{decouple} the runtime from $\poly(n) \cdot \poly(2^t)$ to $\poly(n) + \poly(2^t)$. If the goal is to approximate $U$ to small Frobenius distance, we claim that this runtime is achievable. Indeed, we can design our singular value partitioning step around \cref{lemma:rounded-frobenius-distance} rather than \cref{lem:rounded-diamond-distance}, and this would eliminate the $\exp(t)$ dependence in $\alpha$. Thus the estimation step would take $\poly(n)$ steps and the tomography subroutine in \cref{line:tomography} would be the only step requiring $\exp(t)$ time and samples.

Regarding the constant factors of \cref{alg:learning}, we remark that individual distances between hybrids (e.g. $\diamonddistance{\inputu}{\roundedu} \leq \eps/9$, or $\diamonddistance{\leftgaussian}{\estleftgaussian} \leq \eps/18$) were chosen for ease of presentation.
Choosing smaller errors for, say, the unitary tomography and circuit synthesis steps and relaxing the error requirement between $\estleftgaussian$ and $\leftgaussian$ would yield significantly better constant factors.

Finally, we assert that the exponential dependence on $t$ is necessary. Indeed, consider the class of unitaries of the form $\identity_{n-t} \otimes u$, where $u$ is an arbitrary unitary acting on $t$ modes. This unitary has Gaussian dimension $2(n-t)$. Any algorithm that learns this class in time that is subexponential in $t$ automatically gives a subexponential-time tomography algorithm for unitaries on $t$ modes (or qubits, by the Jordan-Wigner transform). This violates the lower bounds of \cite{baldwin2014quantum,haah2023query,zhao2024learning}. Thus, while it would be unsurprising if the polynomial and constant factors in the runtime of \cref{alg:learning} could be improved, the $\poly(n, 2^t)$ scaling is necessary.

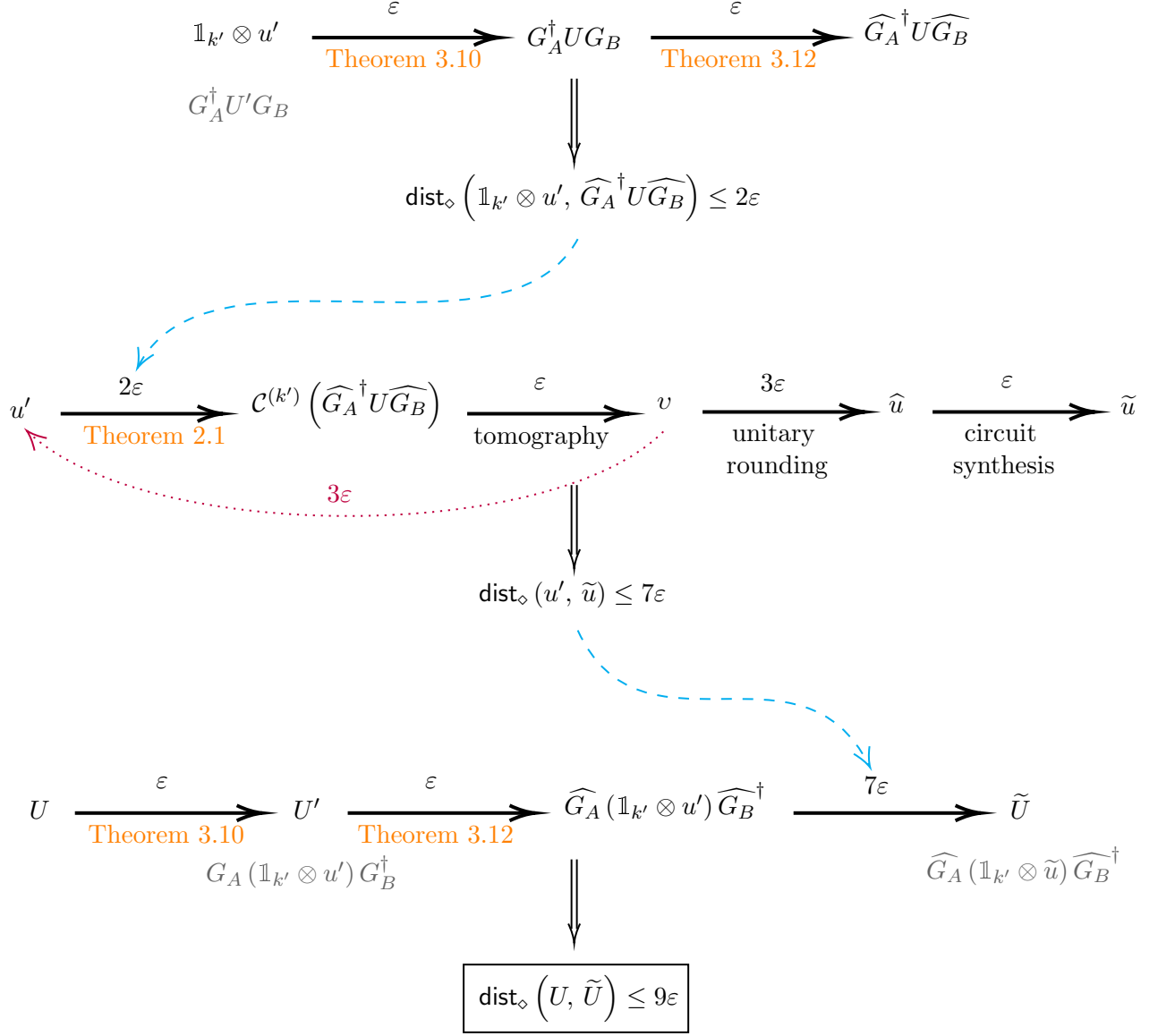
\begin{figure}[H]
    \centering
    \input{diagrams/proof-flowchart.tikz}
    \caption{A flow chart illustrating the hybrid argument in the proof of \cref{thm:main-learning}. 
    Hybrids are drawn with arrows indicating their diamond distance, and relevant lemmas are given below these arrows. Some hybrids have equivalent expressions written below them for convenience.
    Errors have scaled up by a factor of $9$ to aid readability.
    The proof can  be broken into 3 ``layers'': \\
    \textbf{First Layer:} we use the fact that $\inputu$ and $\roundedu$ are close to show that $\identity_{\pk} \otimes \roundedu$, which is $\leftgaussian^\dagger \roundedu \rightgaussian$, is close to $\leftgaussian^\dagger \inputu \rightgaussian$. Then, since $\estleftgaussian$ and $\estrightgaussian$ are close to $\leftgaussian$ and $\rightgaussian$, $\leftgaussian^\dagger \roundedu \rightgaussian$ is close to $\estleftgaussian^\dagger \roundedu \estrightgaussian$. \\ 
    \textbf{Second Layer:} from the previous layer, $\identity_{\pk} \otimes \innerroundedu$ and $\estleftgaussian^\dagger \roundedu \estrightgaussian$ are close, so they remain close after applying $\calC^{(\pk)}$. Then the tomography step \cref{line:tomography} introduces an error $\eps$. This new channel is at most $3\eps$ away from some unitary, as indicated by the \textbf{\color{purple}purple} arrow. Thus the unitary rounding step in \cref{line:unitary-rounding} takes it to a unitary $\innerroundedu$ at most $3\eps$ away. Finally, the circuit synthesis step in \cref{line:circuit synthesis} introduces an error of $\eps$ and we arrive at $\innercircuitu$. \\
    \textbf{Third Layer:} $\roundedu$ is at most $\eps$-far from $\roundedu$, and replacing the $\leftgaussian$ and $\rightgaussian$ in $\roundedu$ with $\estleftgaussian$ and $\estrightgaussian$ introduces an error of $\eps$. Finally, from the second layer, $\innerroundedu$ is at most $7\eps$-far from $\innercircuitu$, concluding the proof.}
    \label{fig:hybrid-argument}
\end{figure}

\section*{Acknowledgements}
The author thanks Scott Aaronson, Jeffrey Champion, Gautam Chandrasekaran, Siddhartha Jain, Harris Junseo Lee, Daniel Liang, Geoffrey Mon, Ojas Parekh, Kevin Thompson, and Brian-Tinh Duc Vu for helpful conversations. Special thanks to Antonio Anna Mele for suggesting the problem to the author, for helpful conversations, and for providing useful feedback. The author is funded by an NSF Graduate Research Fellowship. This work was done in part while the author was visiting the Simons Institute for the Theory of Computing.

\bibliographystyle{alphaurl}
\bibliography{refs}

\appendix

\end{document}

%% file: macros.tex
\usepackage{amsmath, amsthm, amssymb, amsfonts, braket, complexity, enumitem, mathrsfs, tikz, caption, subcaption, xcolor, mathtools, braket, dirtytalk, bbm, float}

\usepackage{tabularx}
\usepackage[utf8]{inputenc}

\usepackage[linesnumbered,ruled,vlined]{algorithm2e}
\SetKwInput{KwInput}{Input}                %
\SetKwInput{KwOutput}{Output}
\SetKwInOut{KwPromise}{Promise}
\SetKwFor{RepTimes}{repeat}{times}{end}

\usepackage[pagebackref, colorlinks,citecolor=blue,linkcolor=orange,urlcolor=cyan,bookmarks=true]{hyperref}
\usepackage[nameinlink,capitalize]{cleveref}

\newcommand{\wt}[1]{\widetilde{#1}}
\newcommand{\wh}[1]{\widehat{#1}}

\newcommand{\F}{\mathbb{F}}
\renewcommand{\hat}{\widehat}

\newcommand{\pk}{k^\prime}
\newcommand{\pt}{t^\prime}

\theoremstyle{plain}
\newtheorem{theorem}{Theorem}[section]
\newtheorem*{theorem*}{Theorem}

\newtheorem{definition}[theorem]{Definition}
\newtheorem{proposition}[theorem]{Proposition}
\newtheorem{lemma}[theorem]{Lemma}

\newtheorem{fact}[theorem]{Fact}

\newtheorem*{remark*}{Remark}
\newtheorem*{question*}{Question}

\newcommand{\eps}{\varepsilon}
\newcommand{\labs}{\left\lvert}
\newcommand{\rabs}{\right\rvert}
\newcommand{\abs}[1]{\labs #1 \rabs}

\renewcommand{\R}{\mathbb{R}}
\newcommand{\majorana}{\gamma}

\newcommand{\newmajorana}{\eta}

\newcommand{\newnewmajorana}{\xi}

\newcommand{\frobeniusdistance}[1]{\func{dist}_{\text{F}}\left(#1\right)}
\newcommand{\diamonddistance}[2]{\func{dist}_{\diamond}\left(#1,\hspace{2pt}#2\right)}

\renewcommand{\Pr}{\mathop{\bf Pr\/}}

\newcommand{\Var}{\mathop{\bf Var\/}}

\newcommand{\tr}{\mathrm{tr}}  
\newcommand{\trace}[1]{\tr\left(#1\right)}
\newcommand{\transpose}{\intercal}

\newcommand{\identity}{\mathbbm{1}}

\newcommand{\calC}{\mathcal{C}}

\newcommand{\calE}{\mathcal{E}}
\newcommand{\calF}{\mathcal{F}}

\newcommand{\calU}{\mathcal{U}}
\newcommand{\calV}{\mathcal{V}}
\newcommand{\calW}{\mathcal{W}}
\newcommand{\calX}{\mathcal{X}}

\renewcommand{\C}{\mathbb{C}}
\renewcommand{\R}{\mathbb{R}}

\newcommand{\norm}[1]{\lVert #1 \rVert}
\newcommand{\frobeniusnorm}[1]{\norm{#1}_{\text{F}}}

\DeclarePairedDelimiter{\ceil}{\lceil}{\rceil}
\DeclarePairedDelimiter{\floor}{\lfloor}{\rfloor}

\DeclareMathOperator{\orthgrp}{O}

\newcommand{\mest}{m_{\text{est}}}

\newcommand{\mtest}{m_{\text{test}}}

\newcommand{\inputu}{U}
\newcommand{\roundedu}{U^\prime}
\newcommand{\inneru}{u}
\newcommand{\innerroundedu}{u^\prime}
\newcommand{\innertomographyu}{\upsilon}
\newcommand{\innerpolaru}{\wh{u}}
\newcommand{\innercircuitu}{\wt{u}}
\newcommand{\outputu}{\wt{U}}
\newcommand{\leftgaussian}{G_A}
\newcommand{\rightgaussian}{G_B}
\newcommand{\estleftgaussian}{\wh{\leftgaussian}}
\newcommand{\estrightgaussian}{\wh{\rightgaussian}}

\definecolor{figRed}{rgb}{0.835,0.188,0.2745}
\definecolor{figYellow}{rgb}{1.0,0.843,0.118}
\definecolor{figYellowText}{rgb}{0.92,0.763,0.038}
\definecolor{figBlue}{rgb}{0.259,0.3647,0.894}

\newcommand{\esterror}{\alpha}
\newcommand{\eigengap}{\zeta}

%% file: diagrams/rounding.tikz
\tikzset{every picture/.style={line width=0.75pt}} %

\begin{tikzpicture}[x=0.75pt,y=0.75pt,yscale=-1,xscale=1]

\draw [line width=1.5]    (298,115) -- (342,115) ;
\draw [line width=1.5]    (347,154) -- (391,154) ;
\draw [line width=1.5]    (396,193) -- (440,193) ;
\draw [line width=1.5]    (450,232) -- (494,232) ;
\draw [line width=1.5]    (499,271) -- (543,271) ;
\draw  [draw opacity=0][fill={rgb, 255:red, 66; green, 93; blue, 229 }  ,fill opacity=1 ] (164,88) .. controls (164,82.48) and (168.48,78) .. (174,78) .. controls (179.52,78) and (184,82.48) .. (184,88) .. controls (184,93.52) and (179.52,98) .. (174,98) .. controls (168.48,98) and (164,93.52) .. (164,88) -- cycle ;
\draw  [draw opacity=0][fill={rgb, 255:red, 66; green, 93; blue, 229 }  ,fill opacity=1 ] (213,90) .. controls (213,84.48) and (217.48,80) .. (223,80) .. controls (228.52,80) and (233,84.48) .. (233,90) .. controls (233,95.52) and (228.52,100) .. (223,100) .. controls (217.48,100) and (213,95.52) .. (213,90) -- cycle ;
\draw  [draw opacity=0][fill={rgb, 255:red, 66; green, 93; blue, 229 }  ,fill opacity=1 ] (262,93) .. controls (262,87.48) and (266.48,83) .. (272,83) .. controls (277.52,83) and (282,87.48) .. (282,93) .. controls (282,98.52) and (277.52,103) .. (272,103) .. controls (266.48,103) and (262,98.52) .. (262,93) -- cycle ;
\draw  [draw opacity=0][fill={rgb, 255:red, 255; green, 215; blue, 30 }  ,fill opacity=1 ] (310,104) .. controls (310,98.48) and (314.48,94) .. (320,94) .. controls (325.52,94) and (330,98.48) .. (330,104) .. controls (330,109.52) and (325.52,114) .. (320,114) .. controls (314.48,114) and (310,109.52) .. (310,104) -- cycle ;
\draw  [draw opacity=0][fill={rgb, 255:red, 255; green, 215; blue, 30 }  ,fill opacity=1 ] (359,131) .. controls (359,125.48) and (363.48,121) .. (369,121) .. controls (374.52,121) and (379,125.48) .. (379,131) .. controls (379,136.52) and (374.52,141) .. (369,141) .. controls (363.48,141) and (359,136.52) .. (359,131) -- cycle ;
\draw  [draw opacity=0][fill={rgb, 255:red, 213; green, 48; blue, 70 }  ,fill opacity=1 ] (463,246) .. controls (463,240.48) and (467.48,236) .. (473,236) .. controls (478.52,236) and (483,240.48) .. (483,246) .. controls (483,251.52) and (478.52,256) .. (473,256) .. controls (467.48,256) and (463,251.52) .. (463,246) -- cycle ;
\draw  [draw opacity=0][fill={rgb, 255:red, 255; green, 215; blue, 30 }  ,fill opacity=1 ] (408,181) .. controls (408,175.48) and (412.48,171) .. (418,171) .. controls (423.52,171) and (428,175.48) .. (428,181) .. controls (428,186.52) and (423.52,191) .. (418,191) .. controls (412.48,191) and (408,186.52) .. (408,181) -- cycle ;
\draw  [draw opacity=0][fill={rgb, 255:red, 213; green, 48; blue, 70 }  ,fill opacity=1 ] (512,259) .. controls (512,253.48) and (516.48,249) .. (522,249) .. controls (527.52,249) and (532,253.48) .. (532,259) .. controls (532,264.52) and (527.52,269) .. (522,269) .. controls (516.48,269) and (512,264.52) .. (512,259) -- cycle ;
\draw    (445,245) -- (445,180) ;
\draw [line width=1.5]  [dash pattern={on 5.63pt off 4.5pt}]  (151,75) -- (542,75) ;
\draw [line width=1.5]  [dash pattern={on 5.63pt off 4.5pt}]  (151,271) -- (543,271) ;
\draw    (439,180) -- (452,180) ;
\draw    (439,245) -- (452,245) ;

\draw (429,201.4) node [anchor=north west][inner sep=0.75pt]  [font=\large]  {$\eigengap $};
\draw (126,67.4) node [anchor=north west][inner sep=0.75pt]    {$1$};
\draw (106,263.4) node [anchor=north west][inner sep=0.75pt]    {$1-\tau $};

\end{tikzpicture}

%% file: diagrams/proof-flowchart.tikz
\tikzset{every picture/.style={line width=0.75pt}} %

\begin{tikzpicture}[x=0.75pt,y=0.75pt,yscale=-1,xscale=1]

\draw [line width=1.5]    (51,462) -- (153,462) ;
\draw [shift={(156,462)}, rotate = 180] [color={rgb, 255:red, 0; green, 0; blue, 0 }  ][line width=1.5]    (14.21,-4.28) .. controls (9.04,-1.82) and (4.3,-0.39) .. (0,0) .. controls (4.3,0.39) and (9.04,1.82) .. (14.21,4.28)   ;
\draw [line width=1.5]    (204,462) -- (254,462) -- (306,462) ;
\draw [shift={(309,462)}, rotate = 180] [color={rgb, 255:red, 0; green, 0; blue, 0 }  ][line width=1.5]    (14.21,-4.28) .. controls (9.04,-1.82) and (4.3,-0.39) .. (0,0) .. controls (4.3,0.39) and (9.04,1.82) .. (14.21,4.28)   ;
\draw [line width=1.5]    (454,462) -- (556,462) ;
\draw [shift={(559,462)}, rotate = 180] [color={rgb, 255:red, 0; green, 0; blue, 0 }  ][line width=1.5]    (14.21,-4.28) .. controls (9.04,-1.82) and (4.3,-0.39) .. (0,0) .. controls (4.3,0.39) and (9.04,1.82) .. (14.21,4.28)   ;
\draw [line width=1.5]    (43,238) -- (131,238) ;
\draw [shift={(134,238)}, rotate = 180] [color={rgb, 255:red, 0; green, 0; blue, 0 }  ][line width=1.5]    (14.21,-4.28) .. controls (9.04,-1.82) and (4.3,-0.39) .. (0,0) .. controls (4.3,0.39) and (9.04,1.82) .. (14.21,4.28)   ;
\draw [line width=1.5]    (271,238) -- (359,238) ;
\draw [shift={(362,238)}, rotate = 180] [color={rgb, 255:red, 0; green, 0; blue, 0 }  ][line width=1.5]    (14.21,-4.28) .. controls (9.04,-1.82) and (4.3,-0.39) .. (0,0) .. controls (4.3,0.39) and (9.04,1.82) .. (14.21,4.28)   ;
\draw [line width=1.5]    (403,237) -- (491,237) ;
\draw [shift={(494,237)}, rotate = 180] [color={rgb, 255:red, 0; green, 0; blue, 0 }  ][line width=1.5]    (14.21,-4.28) .. controls (9.04,-1.82) and (4.3,-0.39) .. (0,0) .. controls (4.3,0.39) and (9.04,1.82) .. (14.21,4.28)   ;
\draw [line width=1.5]    (532,237) -- (589,237) -- (620,237) ;
\draw [shift={(623,237)}, rotate = 180] [color={rgb, 255:red, 0; green, 0; blue, 0 }  ][line width=1.5]    (14.21,-4.28) .. controls (9.04,-1.82) and (4.3,-0.39) .. (0,0) .. controls (4.3,0.39) and (9.04,1.82) .. (14.21,4.28)   ;
\draw [line width=1.5]    (184,27) -- (286,27) ;
\draw [shift={(289,27)}, rotate = 180] [color={rgb, 255:red, 0; green, 0; blue, 0 }  ][line width=1.5]    (14.21,-4.28) .. controls (9.04,-1.82) and (4.3,-0.39) .. (0,0) .. controls (4.3,0.39) and (9.04,1.82) .. (14.21,4.28)   ;
\draw [line width=1.5]    (374,27) -- (476,27) ;
\draw [shift={(479,27)}, rotate = 180] [color={rgb, 255:red, 0; green, 0; blue, 0 }  ][line width=1.5]    (14.21,-4.28) .. controls (9.04,-1.82) and (4.3,-0.39) .. (0,0) .. controls (4.3,0.39) and (9.04,1.82) .. (14.21,4.28)   ;

\draw  [dash pattern={on 0.84pt off 2.51pt}, color=purple]  (381,248) .. controls (317.32,311.68) and (90.29,310.51) .. (24.97,248.44) ;
\draw [shift={(24,247.5)}, rotate = 44.55] [color=purple][line width=0.75]    (10.93,-4.9) .. controls (6.95,-2.3) and (3.31,-0.67) .. (0,0) .. controls (3.31,0.67) and (6.95,2.3) .. (10.93,4.9);

\draw [line width=0.75]    (332.5,50) -- (332.5,89)(329.5,50) -- (329.5,89) ;
\draw [shift={(331,97)}, rotate = 270] [color={rgb, 255:red, 0; green, 0; blue, 0 }  ][line width=0.75]    (10.93,-3.29) .. controls (6.95,-1.4) and (3.31,-0.3) .. (0,0) .. controls (3.31,0.3) and (6.95,1.4) .. (10.93,3.29)   ;
\draw  [dash pattern={on 4.5pt off 4.5pt},color=cyan]  (333,140) .. controls (292.21,214.63) and (125.68,135.3) .. (86.58,211.34) ;
\draw [shift={(86,212.5)}, rotate = 295.97] [color=cyan][line width=0.75]    (10.93,-4.9) .. controls (6.95,-2.3) and (3.31,-0.67) .. (0,0) .. controls (3.31,0.67) and (6.95,2.3) .. (10.93,4.9)   ;
\draw  [dash pattern={on 4.5pt off 4.5pt},color=cyan]  (333,359.5) .. controls (368.82,437.11) and (458.1,359.29) .. (496.43,435.34) ;
\draw [shift={(497,436.5)}, rotate = 244.03] [color=cyan ][line width=0.75]    (10.93,-4.9) .. controls (6.95,-2.3) and (3.31,-0.67) .. (0,0) .. controls (3.31,0.67) and (6.95,2.3) .. (10.93,4.9)   ;
\draw [line width=0.75]    (332.5,278) -- (332.5,317)(329.5,278) -- (329.5,317) ;
\draw [shift={(331,325)}, rotate = 270] [color={rgb, 255:red, 0; green, 0; blue, 0 }  ][line width=0.75]    (10.93,-3.29) .. controls (6.95,-1.4) and (3.31,-0.3) .. (0,0) .. controls (3.31,0.3) and (6.95,1.4) .. (10.93,3.29)   ;
\draw [line width=0.75]    (332.5,488) -- (332.5,527)(329.5,488) -- (329.5,527) ;
\draw [shift={(331,535)}, rotate = 270] [color={rgb, 255:red, 0; green, 0; blue, 0 }  ][line width=0.75]    (10.93,-3.29) .. controls (6.95,-1.4) and (3.31,-0.3) .. (0,0) .. controls (3.31,0.3) and (6.95,1.4) .. (10.93,3.29)   ;

\draw (113,53.4) node [anchor=north west][inner sep=0.75pt]  [font=\normalsize,color={rgb, 255:red, 99; green, 99; blue, 99 }  ,opacity=1 ]  {$G_{A}^{\dagger } \roundedu G_{B}$};
\draw (303,15.4) node [anchor=north west][inner sep=0.75pt]  [font=\normalsize]  {$G_{A}^{\dagger } UG_{B}$};
\draw (115,16.4) node [anchor=north west][inner sep=0.75pt]  [font=\normalsize]  {$\identity_{k^{\prime }} \otimes u^{\prime }$};
\draw (492,9.4) node [anchor=north west][inner sep=0.75pt]  [font=\normalsize]  {$\estleftgaussian^{\dagger } U\estrightgaussian$};
\draw (13,227.4) node [anchor=north west][inner sep=0.75pt]  [font=\normalsize]  {$u^{\prime }$};
\draw (149,217.4) node [anchor=north west][inner sep=0.75pt]  [font=\normalsize]  {$\mathcal{C}^{\left( k^{\prime }\right)}\left(\estleftgaussian^{\dagger } U\estrightgaussian\right)$};
\draw (376,227.4) node [anchor=north west][inner sep=0.75pt]  [font=\normalsize]  {$\innertomographyu $};
\draw (636,226.4) node [anchor=north west][inner sep=0.75pt]  [font=\normalsize]  {$\innercircuitu$};
\draw (506,225.4) node [anchor=north west][inner sep=0.75pt]  [font=\normalsize]  {$\innerpolaru$};
\draw (574,449.4) node [anchor=north west][inner sep=0.75pt]  [font=\normalsize]  {$\outputu$};
\draw (527,479.4) node [anchor=north west][inner sep=0.75pt]  [font=\normalsize,color={rgb, 255:red, 99; green, 99; blue, 99 }  ,opacity=1 ]  {$\estleftgaussian\left(\identity_{k^{\prime }} \otimes \innercircuitu\right)\estrightgaussian^{\dagger }$};
\draw (324,443.4) node [anchor=north west][inner sep=0.75pt]  [font=\normalsize]  {$\estleftgaussian\left(\identity_{k^{\prime }} \otimes u^{\prime }\right)\estrightgaussian^{\dagger }$};
\draw (24,453.4) node [anchor=north west][inner sep=0.75pt]  [font=\normalsize]  {$U$};
\draw (172,452.4) node [anchor=north west][inner sep=0.75pt]  [font=\normalsize]  {$\roundedu$};
\draw (123,484.4) node [anchor=north west][inner sep=0.75pt]  [font=\normalsize,color={rgb, 255:red, 99; green, 99; blue, 99 }  ,opacity=1 ]  {$G_{A}\left(\identity_{k^{\prime }} \otimes u^{\prime }\right) G_{B}^{\dagger }$};
\draw (225,5.4) node [anchor=north west][inner sep=0.75pt]  [font=\normalsize]  {$\eps $};
\draw (417,5.4) node [anchor=north west][inner sep=0.75pt]  [font=\normalsize]  {$\eps $};
\draw (74,216.4) node [anchor=north west][inner sep=0.75pt]  [font=\normalsize]  {$2\eps $};
\draw (307,216.4) node [anchor=north west][inner sep=0.75pt]  [font=\normalsize]  {$\eps $};
\draw (434,214.4) node [anchor=north west][inner sep=0.75pt]  [font=\normalsize]  {$3\eps $};
\draw (569,215.4) node [anchor=north west][inner sep=0.75pt]  [font=\normalsize]  {$\eps $};
\draw (95,440.4) node [anchor=north west][inner sep=0.75pt]  [font=\normalsize]  {$\eps $};
\draw (246,440.4) node [anchor=north west][inner sep=0.75pt]  [font=\normalsize]  {$\eps $};
\draw (492,440.4) node [anchor=north west][inner sep=0.75pt]  [font=\normalsize]  {$7\eps $};
\draw (191,276.4) node [anchor=north west][inner sep=0.75pt]  [font=\normalsize,color=purple]  {$3\eps $};
\draw (273,243) node [anchor=north west][inner sep=0.75pt]   [align=left] {{tomography}};
\draw (415,242) node [anchor=north west][inner sep=0.75pt]   [align=left] {\begin{minipage}[lt]{39.55pt}\setlength\topsep{0pt}
\begin{center}
{unitary}\\{rounding}
\end{center}

\end{minipage}};
\draw (542,242) node [anchor=north west][inner sep=0.75pt]   [align=left] {\begin{minipage}[lt]{40.13pt}\setlength\topsep{0pt}
\begin{center}
{circuit}\\{synthesis}
\end{center}

\end{minipage}};
\draw (54,243) node [anchor=north west][inner sep=0.75pt]   [align=left] {\Cref{fact:diamond-distance-contraction}};
\draw (190,32) node [anchor=north west][inner sep=0.75pt]   [align=left] {\Cref{lem:rounded-diamond-distance}};
\draw (378,32) node [anchor=north west][inner sep=0.75pt]   [align=left] {\Cref{lem:svd-gaussian-error-bound}};
\draw (57,467) node [anchor=north west][inner sep=0.75pt]   [align=left] {\Cref{lem:rounded-diamond-distance}};
\draw (208,467) node [anchor=north west][inner sep=0.75pt]   [align=left] {\Cref{lem:svd-gaussian-error-bound}};
\draw (235,101.4) node [anchor=north west][inner sep=0.75pt]    {$\diamonddistance{\identity_{k^{\prime }} \otimes u^{\prime }}{\estleftgaussian^{\dagger } U\estrightgaussian} \leq 2\eps $};
\draw (276,330.4) node [anchor=north west][inner sep=0.75pt]    {$\diamonddistance{\innerroundedu }{\innercircuitu} \leq 7\eps $};
\draw    (269,547) -- (395,547) -- (395,585) -- (269,585) -- cycle  ;
\draw (276,551.4) node [anchor=north west][inner sep=0.75pt]    {$\diamonddistance{\inputu}{\outputu} \leq 9\eps $};

\end{tikzpicture}